\documentclass[11pt]{article}

\usepackage{amsmath, amssymb, amsthm, amsfonts, latexsym, graphicx, color, tikz}
\usepackage{arydshln}
\usepackage{authblk}

\usetikzlibrary{matrix,calc,fit}

\newlength{\actualtopmargin}
\newlength{\actualsidemargin}
\theoremstyle{plain}
\newtheorem{theorem}{Theorem}

\theoremstyle{definition}
\newtheorem{definition}[theorem]{Definition}

\theoremstyle{remark}

\theoremstyle{plain}
\newtheorem*{theorem*}{Theorem}
\newtheorem*{lemma*}{Lemma}
\newtheorem*{corollary*}{Corollary}
\newtheorem*{proposition*}{Proposition}
\newtheorem*{claim*}{Claim}

\newcommand\highlight[2][]{\tikz[overlay]\node[fill=red,inner sep=2pt, anchor=text, rectangle, rounded corners=1.4mm,#1] {\color{white}#2};\phantom{#2}}

\newcommand{\ii}{\mathbb{I}}

\newcommand{\poly}{\mathrm{poly}}

\newcommand{\oo}[1]{\Theta\left(#1\right)} 
\newcommand{\oOmega}[1]{\Omega\left(#1\right)} 
\newcommand{\oOh}[1]{O\left(#1\right)} 
\newcommand{\bra}[1]{\langle #1 \vert}

\newcommand{\ket}[1]{\vert #1 \rangle}

\newcommand{\braket}[2]{\langle #1 \vert #2 \rangle}

\begin{document}
\title{\Large \textbf{The Feynman-Kitaev computer's clock: bias, gaps, idling and pulse tuning}}
\author[1]{Libor Caha\footnote{libor.caha@savba.sk}}
\affil[1,3]{Research Center for Quantum Information, Institute of Physics, Slovak Academy of Sciences, D\'ubravsk\'a cesta 9, 845 11 Bratislava, Slovakia}
\author[2]{Zeph Landau}
\affil[2]{Department of Electrical Engineering and Computer Sciences, University of California, Berkeley, CA 94720, U.S.A.}
\author[3]{Daniel Nagaj\footnote{daniel.nagaj@savba.sk}}
\maketitle
\vspace{-5mm}


\begin{abstract}
	We present a collection of results about the clock in Feynman's computer construction and Kitaev's Local Hamiltonian problem. First,
	by analyzing the spectra of quantum walks on a line with varying endpoint terms,
	we find a better lower bound on the gap of the Feynman Hamiltonian, which translates into a less strict promise gap requirement for the QMA-complete Local Hamiltonian problem. We also translate this result into the language of adiabatic quantum computation. Second, introducing an idling clock construction with a large state space but fast Cesaro mixing, we provide a way for achieving an arbitrarily high success probability of computation with Feynman's computer with only a logarithmic increase in the number of clock qubits. Finally, we tune and thus improve the costs (locality, gap scaling) of implementing a (pulse) clock with a single excitation.
\end{abstract}

\maketitle

\section{Introduction}

The need to describe and find the properties of many-body systems in quantum physics has lead to a large collection of interesting computational problems.
Some are easy for classical computers \cite{Q2SAT, Landau:2015aa}, some efficiently verifiable on a quantum computer \cite{8state, Schuch:2009aa, AdamBookatzReview}, and others even undecidable \cite{Cubitt:2015aa}.
The development of numerical methods for these problems is a field to itself with exciting new developments motivated by quantum information \cite{SCHOLLWOCK201196, PEPS}.
On the other hand, the goal of quantum Hamiltonian complexity \cite{TCS-066} is to theoretically understand the universal power of models of computation based on local Hamiltonians \cite{MontanaroCubittPiddock, UniversalAdiabaticCircuitHamiltonian},
as well as to characterize the computational complexity of Hamiltonian-based optimization \cite{Quantum3SAT}, rewriting \cite{JanzingWocjan}, connectivity \cite{gscon}, degeneracy \cite{UniqueQWitness}, sampling \cite{Aaronson:2011:CCL:1993636.1993682} and other types of problems.

Some of the questions involve static properties of the Hamiltonians describing the system. For example, the existence of eigenstates with a certain energy bound \cite{KitaevBook}, the behavior of quantum correlations \cite{MovaShor}, or the possibility of finding parent Hamiltonians given an eigenstate \cite{2018arXiv180207827G}.
Other questions involve dynamics, asking about computational and universality and simulation \cite{MontanaroCubittPiddock},
or the possibilities of state preparation \cite{AdiabaticStateGeneration, PEPSprepare}.
The roots of some of these questions can be traced back to Feynman, who devised a computational model based on unitary evolution with a fixed quantum mechanical Hamiltonian \cite{FeynmanQMC}.
There is a crucial difference from classical computation -- one can no longer efficiently read out and store (copy) the state of the system at any point of the computation. Feynman's computer works with superpositions over snapshots of the computation, with unitary transformations according to Schroedinger evolution with a particular Hamiltonian, on a system with two registers -- clock and data. There are many ways one can implement this by a local Hamiltonian, depending on the intended application. Our goal is to improve several of these techniques.

This paper is a collection of results about clocks for quantum complexity constructions, tied together by quantum walk techniques. We utilize local interactions to construct the clock register and couple it to what happens in the data registers, with improved efficiency (fewer required steps), complexity requirements (gap promise), success probability (spatially efficient probabilistic computation with a tunable success rate), and locality of interactions (few-body terms).

In Section~\ref{sec:LH}, we start with a review Feynman's ideas of computing with a Hamiltonian, Kitaev's local Hamiltonian problem, and the clocks that they use.
We then present our first result
about a class of Hamiltonians describing a clock biased towards one end of the computation in Section~\ref{sec:biasedwalk}. Relying on a mapping to quantum walks on a line with endpoint self-loops, we show that these Hamiltonians are gapped.

Second, in Section~\ref{sec:KitaevGap} we apply what we learned about clocks with biased ends to improve the promise bound for Kitaev's local Hamiltonian problem. Note that this bound has been recently independently similarly improved from $\oOmega{N^{-3}}$ to $\oOmega{N^{-2}}$ by Bausch \& Crosson \cite{BauschCrossonGap} who have also looked at tridiagonal Hamiltonians, but used a Markov chain mixing technique instead of quantum walks. They also showed that this bound is tight for any clock whose Hamiltonian is tridiagonal in the time-register basis.

Our third, negative result is an analysis of the efficiency of universal computation by adiabatic evolution in Section~\ref{sec:adiabatic}, relying on what we learned about gaps of biased clocks. We find that adiabatic quantum computation with standard Hamiltonians does not yield a natural quadratic speedup over quantum computation with a static quantum walk Hamiltonian and mixing.

Fourth, we present two ways of doing nothing (idling the engine) to improve the success probabilities for quantum computation with local Hamiltonians in Section~\ref{sec:idling}. Most importantly, we do it efficiently (with a sublinear increase in the number of used qubits). The first construction is designed for static applications in complexity, increasing the overlap of the ground state of a local Hamiltonian with a state containing the result of a computation. The second method is less efficient, but usable in dynamical constructions (i.e. for building a computer).

We envision the use of these results in quantum Hamiltonian complexity applications -- giving one a better starting position for gap amplification, tighter bounds and better understanding of commonly used quantum walks on a line with boundary terms, as well as two methods to efficiently tune the success probability of a computation efficiently in terms of space, running time and locality.


\section{The Feynman-Kitaev Computer, Clocks and Gaps}
\label{sec:LH}

In this Section we briefly review universal computation with the Feynman Hamiltonian and Kitaev's QMA-complete Local Hamiltonian problem.
Quantum computation is usually viewed in terms of the circuit model, with a large unitary circuit on $n$ qubits decomposed into a sequence of unitary gates, each acting on a few qubits. It is also possible to evaluate some of the gates in parallel.
An equivalent universal quantum computation formulation is possible using time independent Hamiltonians. There, an initial state unitarily evolves according to the Schr\"{o}dinger equation for some Hamiltonian $H$ built from local interaction terms, each acting on a few particles.
We will first show such a construction. Second, we will show how it can be translated to a static construction, where a ground state of a local Hamiltonian encodes the progress of a computation with a unitary circuit.

Let us clear some notation and labeling issues. In this paper, we talk about $k$-local Hamiltonians, built from terms act nontrivially only on $k$ particles. This does not necessarily imply {\em geometric locality}, which would means that the particles are also spatially close, e.g. for nearest-neighbor interactions on a lattice.
We also use a simplified notation for operators $O$ acting in a larger Hilbert space $S$, but nontrivially only on a smaller subspace $A$ as $O = A_A \otimes \ii_{S-A}$, writing just $A_A$ instead of the full expression. For example, we will denote the projector $\ii_{1,2} \otimes \ket{00}\bra{00}_{3,4}\otimes \ii_{5,\dots,N}$ by the shorter and more readable $\ket{00}\bra{00}_{3,4}$ acting on the subsystems $3$ and $4$, and implicitly understanding that it acts on a larger Hilbert space with $N$ subsystems.
Finally, we utilize the standard asymptotic (big-O) notation (see e.g. \cite{Cormen:2009:IAT:1614191}), where
\begin{enumerate}
\item $f(n)=\oOh{g(n)}$ means that $f(n)$ is asymptotically bounded from above by $g(n)$, i.e. there exist constants $c, n_0>0$ such that for all $n\geq n_0$ we have $0\leq f(n) \leq c g(n)$.
\item $f(n)=\oOmega{g(n)}$ means that $f(n)$ is asymptotically bounded from below by $g(n)$, i.e. there exist constants $c, n_0>0$ such that for all $n\geq n_0$ we have $0\leq c g(n)\leq f(n)$.
\item $f(n)=\oo{g(n)}$ means that $f(n)$ is asymptotically bounded by $g(n)$ both from above and below, i.e. it obeys $f(n)=\oOh{g(n)}$ and $f(n)=\oOmega{g(n)}$ at the same time.
\item $f(n)=o(g(n))$ means that $f(n)$ is asymptotically dominated by $g(n)$, i.e. for any constant $c>0$ there exists a constant $n_0>0$ such that for all $n\geq n_0$ we have $0\leq f(n) < c g(n)$.
\end{enumerate}

\subsection{A dynamical construction: Feynman's computer}
We now present Feynman's construction for performing a unitary computation by evolving with a static Hamiltonian.
Consider a quantum circuit $U = U_N U_{N-1}\dots U_2 U_1$ composed of $N$ gates.
Our playground will be a Hilbert space made from a {\em clock} register holding $N+1$ possible states $\ket{0},\dots, \ket{N}$ labeling the progress of the computation, and a {\em data} register that will hold the qubits we want to compute on:
\begin{align}
	\mathcal{H}=\mathcal{H}_{\textrm{clock}} \otimes \mathcal{H}_{\textrm{data}}.
	\label{2registers}
\end{align}
When we evolve an initial state $\ket{\psi_0^{0}} =\ket{0}_{\textrm{clock}}\otimes \ket{0\cdots 0}_{\textrm{data}}$
with Feynman's Hamiltonian
\begin{align}
	H_{\textrm{F}} &= \sum_{t=1}^{N}
		\left(
		\ket{t}\bra{t-1}_{\textrm{clock}}\otimes U_t
		+
		\ket{t-1}\bra{t}_{\textrm{clock}}\otimes U_t^\dagger
		\right), \label{Hf}
\end{align}
the resulting state will live in the space
\begin{align}
	\mathcal{H}_0 = \textrm{span}\left\{
		\ket{\psi_t^{0}} = \ket{t}_{\textrm{clock}}\otimes
		\left(U_t \dots U_1 \ket{0\dots 0}_{\textrm{data}}\right), \, t=0,\dots,N
		\right\}. \label{Hpsi}
\end{align}
Observe also that $H_{\textrm{F}} \ket{\psi_t^{0}} =\ket{\psi_{t-1}^{0}}+ \ket{\psi_{t+1}^{0}}$, so the restriction $H_{\textrm{F}}\big|_{\mathcal{H}_0}$ is the Hamiltonian of a continuous-time quantum walk on a line \cite{QWalksActaPhysica} of states $\ket{\psi_t^{0}}$.
Using quantum walk techniques, we can show that when we evolve the initial state $\ket{\psi_0^{0}}$ for a time randomly chosen between $0$ and $\oo{N^2}$, and measure the clock register,
with probability $\oo{N^{-1}}$ we will obtain the state $\ket{N}_{\textrm{clock}}$, and thus $U_N \dots U_1 \ket{0\dots 0}$ (the result of the circuit $U$ applied to the initial state $\ket{0\dots 0}$) in the data register. Therefore, evolution with Feynman's Hamiltonian is a universal quantum computer.

Below, in Section \ref{sec:clocks}, we show that Feynman's computer can be built from local terms, by choosing a local implementation of the clock register states and the Hamiltonian terms inducing transitions terms between the states.

\subsection{A static construction: Kitaev's Hamiltonian}

Here, we show how to turn Feyman's dynamic construction (evolve with $H_{\textrm{F}}$ and measure, obtaining the result of a quantum computation) to a static construction, where the ground state is a {\em history state}, encoding the progress of a quantum computation.

Let us consider the Hilbert space \eqref{2registers} with two (clock/data) registers.
We call
\begin{align}
	\ket{\psi_{\textrm{hist}}^{\varphi}} = \frac{1}{\sqrt{N+1}} \sum_{t=0}^{N} \ket{\psi_t^{\varphi}}
	= \frac{1}{\sqrt{N+1}} \sum_{t=0}^{N} \ket{t}_{\textrm{clock}}
			\otimes \underbrace{U_t U_{t-1} \dots U_1 \ket{\varphi}}_{\ket{\varphi_t}_{\textrm{data}}}
		\label{history}
\end{align}
the {\em history state} of the computation of the circuit $U$ on the initial data register state $\ket{\varphi}$.
Kitaev constructed a Hamiltonian whose ground states have the form of such history states -- uniform superpositions of successive states of the computation of the circuit $U$, along with a {\em clock} register labeling the progress of the computation. This {\em propagation checking} Hamiltonian\footnote{To avoid repeating many $\frac{1}{2}$'s later on, in this paper we choose to omit the usual constant prefactor $\frac{1}{2}$ in $H_{\textrm{prop}}$.} is
\begin{align}
	H_{\textrm{prop}} &= \sum_{t=1}^{N}
		\left(
		\left(\ket{t-1}\bra{t-1} 	+ \ket{t}\bra{t}\right)_{\textrm{clock}}\otimes \ii_{\textrm{data}}
		-
		\ket{t}\bra{t-1}_{\textrm{clock}}\otimes U_t
		-
		\ket{t-1}\bra{t}_{\textrm{clock}}\otimes U_t^\dagger
		\right), \label{Hprop}
\end{align}
built from projector terms and Feynman's Hamiltonian \eqref{Hf}.
In Section \ref{sec:clocks}, we will see that $H_{\textrm{prop}}|_{\mathcal{H}_{\varphi}}$, restricted to the Hilbert space $\mathcal{H}_{\varphi}$  \eqref{Hpsi} spanned by states $\ket{\psi_t^{\varphi}}$ \eqref{history}, is a quantum walk on a line with self-loops.

Kitaev then used it to give a QMA-complete problem, the {\em Local Hamiltonian} \cite{KitaevBook}.
He showed how to construct a Hamiltonian with a ground state energy below some bound only if there exists an initial state $\ket{\varphi}$, for which the output qubit of the state $U\ket{\varphi}$ is $\ket{1}$ with high probability. If there is no such state $\ket{\varphi}$, the ground state energy is above some bound. This is one reason behind why determining with high precision the ground state energy of local Hamiltonians is difficult.

This Hamiltonian is made from four terms:
\begin{align}
	H_{\textrm{K}} = H_{\textrm{prop}} + H_{\textrm{init}} + H_{\textrm{out}} + H_{\textrm{clock}}.
	\label{HKitaev}
\end{align}
Each of them ``checks'' some property of the state, giving lower energies to states that have this property. $H_{\textrm{prop}}$ prefers proper propagation of the computation, ensuring that the low energy states of \eqref{HKitaev} are close in form to the history states \eqref{history}.
The second term in \eqref{HKitaev} checks the proper initialization of ancilla qubits at the start of the computation. $H_{\textrm{out}}$ looks at whether the result of the computation is ``accept'', i.e. whether the state of the designated output data qubit (labelled {\em out}) is $\ket{1}$ at the end of the computation, i.e. when the clock register reads $\ket{N}$. In detail,
\begin{align}
	H_{\textrm{init}} &= \sum_{\textrm{ancillas }a} \ket{0}\bra{0}_{\textrm{clock}} \otimes \ket{1}\bra{1}_{a}, \qquad
		H_{\textrm{out}} =  \ket{N}\bra{N}_{\textrm{clock}} \otimes \ket{0}\bra{0}_{\textrm{out}},
\end{align}
The final term in \eqref{HKitaev} is a clock-checking Hamiltonian,
checking the proper form of the states in the clock register.
The particular implementation of the clock register
 and its interaction with the data \eqref{Hprop} by a local Hamiltonian
is crucial for making the Kitaev Hamiltonian local.
In the next Section
we will see how this can be done, e.g. by Kitaev's original domain-wall (unary) clock (see Section~\ref{sec:dw}).

\subsection{Clock constructions}
\label{sec:clocks}

The basic building block for Feynman's computer (and Kitaev's Local Hamiltonian construction) is a clock -- a register with $N+1$ possible logical states $\ket{0},\dots,\ket{N}$, denoting the linear progress of a computation. Originally, Feynman envisioned it being a hopping pointer particle.
Here we will look at this construction and other options, their properties, and ways to make them local.

Note that one could also construct clocks with a nonlinear progression of states, without unique forward and backward transitions. In recent quantum complexity results \cite{Quantum3SAT}, we have seen the combinations of several clock registers, blind alley transitions, railroad-switching paths and path noncommutatitivity, amongst other ideas. However, there are still interesting things to be learned about the basic linear approaches and their relationship to quantum walks, as we will show below.

The clock for Feynman's computer can be realized by a hopping Hamiltonian (a quantum walk on a line):
\begin{align}
	H_{N}^{\textrm{walk}} &= - \sum_{t=0}^{N-1}
	\left(
	\ket{t+1}\bra{t}+\ket{t}\bra{t+1}
	\right) \label{Hwalk}
\end{align}
acting on a Hilbert space of size $N+1$, spanned by the states $\ket{t}$ for $t=0,\dots,N$. We choose the minus sign in front of the Hamiltonian for convenience, so that later the low-energy states have positive amplitudes. This Hamiltonian is the negative of the adjacency-matrix of a line. Its eigenvectors are then combinations of plane waves with certain momenta, analyzed in detail in Section~\ref{sec:LR}. The gap (difference of two lowest energies) of such Hamiltonians scales as $\oo{N^{-2}}$.

In Kitaev's construction and followup work \cite{OliveiraTerhal,new3local}, the clock Hamiltonian is usually written as a sum of projectors. It is related to \eqref{Hwalk} as an adjacency matrix is related to a Laplacian of a graph with edges corresponding to possible clock transitions.
For each transition $\ket{t} \leftrightarrow \ket{t+1}$ in \eqref{Hwalk}, the projector $\frac{1}{2}\left(\ket{t}-\ket{t+1}\right)\left(\bra{t}-\bra{t+1}\right)$ energetically prefers a uniform superposition of these states, i.e. $\frac{1}{\sqrt{2}}\left(\ket{t}+\ket{t+1}\right)$.
Let us write down the sum of these projectors, and omit the $\frac{1}{2}$ prefactor for simplicity.
We call this Hamiltonian the {\em Laplacian quantum walk}:
\begin{align}
	H_{N}^{\textrm{L}} &= \sum_{t=0}^{N-1}
	\left(\ket{t}-\ket{t+1}\right)\left(\bra{t}-\bra{t+1}\right). \label{HLap}
\end{align}
It is a frustration-free sum of positive semidefinite terms with a unique, 0-energy ground state -- the uniform superposition $\frac{1}{\sqrt{N+1}}\sum_{t=0}^{N} \ket{t}$. We can also view this ground state as a history state \eqref{history} without a data register, for a circuit made out of identity gates.
Expanding \eqref{HLap}, we find its relationship to \eqref{Hwalk}:
\begin{align}
	H_{N}^{\textrm{L}} &=
	\ket{0}\bra{0} + 2\sum_{t=1}^{N-1} \ket{t}\bra{t} + \ket{N}\bra{N} -
	\sum_{t=0}^{N-1}
	\left(\ket{t+1}\bra{t} + \ket{t}\bra{t+1}\right) \nonumber\\
	&= 2\ii -
	\ket{0}\bra{0} - \ket{N}\bra{N} +H_{N}^{\textrm{walk}}.
	\label{HLapends}
\end{align}
It is the Laplacian matrix for a line graph of length $N+1$.
It can also be interpreted as a shifted quantum walk on a line \eqref{Hwalk} with endpoint projectors. We analyze such walks in Section~\ref{sec:lineLR}.

Let us now look at how the clock Hamiltonians \eqref{Hwalk} and \eqref{HLap} can be implemented in spin systems, in particular, in spin chains with nearest-neighbor or next-nearest-neigbor interactions.

\subsubsection{The pulse clock: an excitation hopping on a line}
\label{sec:pulse}

One can use Hamiltonians with 2-local interactions to implement a linear clock from the previous Section. One option is to model \eqref{Hwalk} by the hopping of a single excitation in a spin-$\frac{1}{2}$ chain of length $N+1$.
The states $\ket{t}$ for $H_{N}^{\textrm{walk}}$ correspond to spin chain states $\ket{0 \cdots 0 1_x 0 \cdots 0}$ with the $\ket{1}$ at position $x=t+1$.
In Feynman's computer, the position of the $\ket{1}$ (the excitation, the pointer) then measures the progress of the computation.
The nearest-neighbor spin chain Hamiltonian reads
\begin{align}
	H_{N}^{\textrm{pulse}} &= - \sum_{x=1}^{N}
	\left(\ket{01}\bra{10}+\ket{10}\bra{01}\right)_{x,x+1}. \label{Hpulse}
\end{align}
For completeness, a Laplacian walk version of \eqref{Hpulse} would read
\begin{align}
	H_{N}^{\textrm{pulse,L}} &= \sum_{x=1}^{N}
	\left(\ket{01}-\ket{10}\right)\left(\bra{01}-\bra{10}\right)_{x,x+1}. \label{Hpulseproj}
\end{align}

To ensure that our pulse clock (in a spin chain) works in the {\em good} subspace with a single excitation, we only need to initialize it this way. Observe that the Hamiltonians \eqref{Hpulse}, \eqref{Hpulseproj} keep the number of 1's in the chain invariant. We show another option in Section~\ref{sec:tuning}: adding a precisely tuned local Hamiltonian that prefers the single-excitation subspace over others (strongly hating neighboring 11's, while locally weakly preferring 1's over 0's).

The properties of the eigenvectors and the spectra of \eqref{Hpulse}, \eqref{Hpulseproj} restricted to the good subspace are the same as the properties of the quantum walk Hamiltonians \eqref{Hwalk}, \eqref{HLap}. Note that these are also well known in condensed matter physics.
The Hamiltonian of the walk on a line \eqref{Hwalk} implemented by the pulse clock \eqref{Hpulse} can be mapped to to the 1-excitation sector of the ferromagnetic XX-model spin chain
\begin{align}
		-\sum_{x} \left( \ket{01}\bra{10}+ \ket{10}\bra{01}\right)_{x,x+1}
		= -\frac{1}{2}\sum_{x} X_x X_{x+1} \left( \ii - Z_x Z_{x+1}\right)
		= -\frac{1}{2}\sum_{x} \left( X_x X_{x+1} + Y_x Y_{x+1}\right).
\end{align}
On the other hand, the behavior of the projector Hamiltonian \eqref{Hpulseproj} can be mapped to the 1-excitation sector of the ferromagnetic Heisenberg (XXX model) chain.
\begin{align}
		\sum_{x} \left( \ket{10}-\ket{01}\right)\left(\bra{10}-\bra{01}\right)_{x,x+1}
		&= \sum_{x}
			\left(\ii - X_x X_{x+1}\right) \frac{1}{2}\left( \ii - Z_x Z_{x+1}\right)\nonumber \\
		&= \frac{1}{2} \ii
			- \frac{1}{2} \sum_{x} \left( X_x X_{x+1} + Y_x Y_{x+1} + Z_x Z_{x+1}\right).
\end{align}

\subsubsection{The domain wall (unary) clock}
\label{sec:dw}

The second option for a clock is the {\em domain wall} (unary) clock.
This was the version of the clock used in Kitaev's 5-local Hamiltonian -- the clock itself is 3-local, and adding 2-qubit gates to build \eqref{Hprop} makes it 5-local.
It involves a progression of states with a single domain wall, like
	$\ket{100000}, \ket{110000}, \ket{111000}, \ket{111100}, \ket{111110}.$
It can be implemented by a 3-local (next-nearest-neighbor) Hamiltonian on a spin-$\frac{1}{2}$ chain of length $N+2$. First, we can model the hopping \eqref{Hwalk} as
\begin{align}
	H_{N}^{\textrm{dw,walk}} = - \sum_{x=1}^{N} \left(
		\ket{110}\bra{100}
		+ \ket{100}\bra{110}
		\right)_{x,x+1,x+2}, \label{Hunaryhop}
\end{align}
while being restricted to the {\em good} subspace spanned by the states $\ket{1\cdots 1_x 0\cdots 0}$ with $x=t+1$ ones corresponding to the clock state $\ket{t}$. We can easily construct a clock-checking Hamiltonian
\begin{align}
H^{\textrm{dw}}_{\textrm{clock-check}} = \sum_{x=1}^{N+1} \ket{01}\bra{01}_{x,x+1}
		+ \ket{0}\bra{0}_1
		+ \ket{1}\bra{1}_{N+2}, \label{Hdwcc}
\end{align}
that energetically favors the good subspace, because only the single-domain-wall states have no neighboring $01$'s, start with a $1$, and end with a $0$.

Again, we can also write down a Laplacian version \eqref{HLap} of the domain-wall clock.
\begin{align}
	H_{N}^{\textrm{dw,L}} = \sum_{x=1}^{N} \left( \ket{100} - \ket{110} \right)
	\left( \bra{100} - \bra{110} \right)_{x,x+1,x+2}. \label{Hunary}
\end{align}
Adding clock-checking \eqref{Hdwcc}, we find that
the positive semidefinite Hamiltonian $H_{N}^{\textrm{dw,L}}+H^{\textrm{dw}}_{\textrm{clock-check}}$
has a unique, frustration-free (annihilated by all projector terms), zero-energy ground state
	$
	\ket{\psi_{\textrm{dw}}} = \frac{1}{\sqrt{N+1}} \sum_{x=1}^{N+1} \ket{1\cdots 1_x 0 \cdots 0}
	$.

The Hamiltonians described in this Section do not introduce or delete domain walls. The Hilbert space thus splits into the invariant {\em good} subspace spanned by states with a single domain wall, and other invariant subspaces. In those, all states have energy at least a constant $E\geq 1$, because each such state is ``detected'' by at least one of the clock-checking terms in \eqref{Hdwcc}.

By construction, \eqref{Hunary} restricted to the good subspace becomes \eqref{HLap},
a rescaled and shifted quantum walk \eqref{Hwalk} on a line of $N+1$ states, with extra endpoint projector terms that can be interpreted as endpoint loops. The eigenvectors of \eqref{HLap} are again combinations of plane waves. We show in Section \ref{sec:lineLR} that the gap (the difference between the two lowest eigenvalues) of this Hamiltonian again scales like $\oOmega{N^{-2}}$.

The pulse and domain-wall clocks are the most ``vanilla'' constructions, where we have complete understanding of the invariant subspaces, the good subspace, the eigenvectors and eigenvalues, the dynamics as well as the gap. Let us turn to slightly more complicated cases, modifying the walk on a line of states. First, we will introduce a bias towards one side in Section~\ref{sec:biasedwalk}, and then analyze what happens in a more general case with varying strength of attraction/repulsion at the endpoints in Section~\ref{sec:lineLR}.

\subsection{Clocks as walks on a line}
\label{sec:biasedwalk}

We now delve into a more general investigation of clocks that correspond to walks on a line. We would like to get a faster computation, larger success probability, larger overlaps with the completed computation, or better spectral properties for Feynman's computer or Kitaev's construction. For this purpose, we will first investigate a walk biased to one side, and then walks with endpoint projectors with varying strength. It turns such walks (clocks) are specific tridiagonal matrices, whose spectral properties we can analyze. The main application we find is an improved lower bound on the required precision for Kitaev's QMA-complete problem Local Hamiltonian.

\begin{figure}%
\begin{center}
\includegraphics[width=16cm]{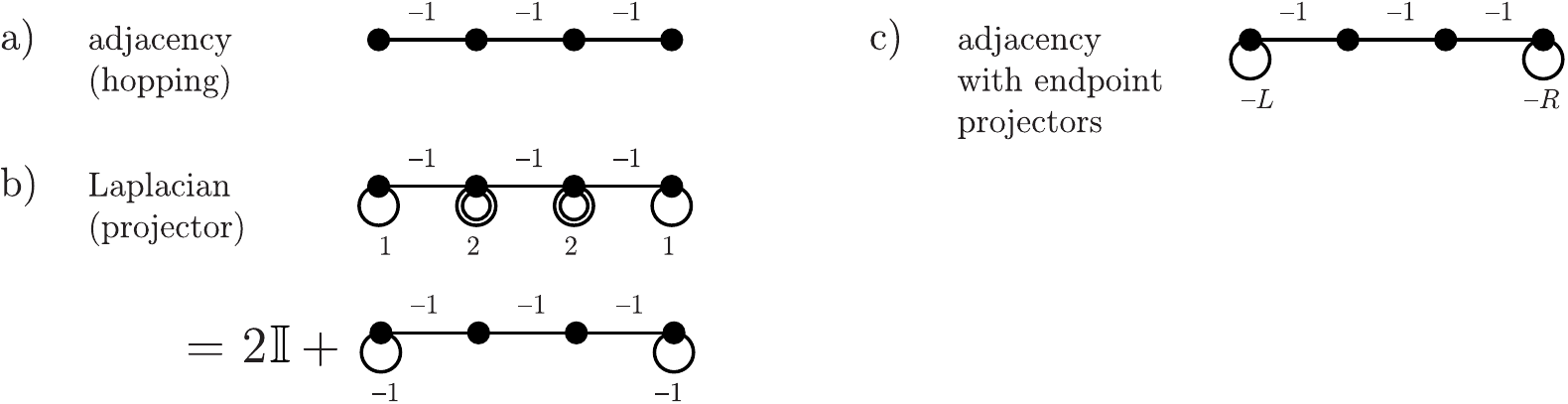}%
\caption{Quantum walks on a line of length $N$. a) The quantum walk on a line (hopping) Hamiltonian $H_N^{\textrm{walk}}$ \eqref{Hwalk} is the negative of the adjacency matrix.
b) The Laplacian walk $H_N^{\textrm{L}}$ \eqref{HLap} includes a self loop on each vertex for each outgoing edge.
c) A more general version $H_N^{(L,R)}$ \eqref{HwalkLR} parametrized by a pair $L,R$ includes endpoint projectors (loops) $-L\ket{0}\bra{0}$, $-R\ket{N}\bra{N}$.
}%
\label{fig:walktypes}%
\end{center}
\end{figure}

In Figure~\ref{fig:walktypes} we depict the types of walks on a line that we investigate here. First, we depict the adjacency walk (hopping) from \eqref{Hwalk}, then the Laplacian walk (projector) from \eqref{HLapends}, and then a generalization with variable endpoint loop projector terms that we will analyze below \eqref{HwalkLR}. We are looking into this particular generalization because its form encompasses many interesting cases. For example, such Hamiltonians appear in the quantum walk algorithm for traversing randomly glued trees by Childs et al. \cite{ChildsTrees} or in adiabatic quantum computation applications
\cite{UniversalAdiabaticCircuitHamiltonian}, \cite{NagajKieferovaSomma}, where one interpolates between endpoint projector terms and a quantum walk/Kitaev's computer term.

The simplest generalization is to
start with \eqref{HLap} and bias the clock towards one side by a parameter $B>1$:
\begin{align}
	H^{\textrm{bias}}_{N} &= 
	\sum_{x=0}^{N-1}
	\left(B\ket{x}-\ket{x+1}\right)
	\left(B\bra{x}-\bra{x+1}\right).  \label{biasedclock}
\end{align}
It is a sum of positive semidefinite terms (each of them is proportional to a projector), and thus positive semidefinite. Instead of the uniform superposition, it energetically prefers superpositions of the form $\ket{x}+B\ket{x+1}$.
The ground state will thus involve an exponential increase in amplitude towards one of the endpoints.
\begin{align}
	\ket{\psi_{\textrm{bias}}} = \frac{1}{\sqrt{\sum_{x=0}^{N} B^{2x} }}
	\sum_{x=0}^{N} B^x \ket{x}. \label{expstate}
\end{align}
This ground state is unique and fully annihilated by $H_{N}^{\textrm{bias}}$, making $H^{\textrm{bias}}_{N}$ frustration-free.
We can rewrite the new Hamiltonian \eqref{biasedclock} as
\begin{align}
	H^{\textrm{bias}}_{N} &= \left(1+B^2\right) \ii + B\left(- \frac{1}{B}\ket{0}\bra{0} - B \ket{N}\bra{N}
	 +  H_{N}^{\textrm{walk}}
	\right), \label{rewriteHbias}
\end{align}
and interpret it as a shifted and rescaled quantum walk on a line \eqref{Hwalk} with self-loops of magnitude $\frac{1}{B}$ and $B$. This is a special case of a walk on a line with general self-loops at the end, which we will analyze in Section~\ref{sec:lineLR} and prove that its gap is a constant.
Looking at the form of \eqref{rewriteHbias}, another generalization with fully tunable endpoint loops comes naturally.

\subsubsection{Walking on a line with tunable endpoint loops}
\label{sec:lineLR}

Let us now analyze a more general case: a quantum walk on a line with $N$ links and endpoint self-loops of constant strength $L$ and $R$.
Our Hamiltonian now has the form
\begin{align}
	H^{(L,R)}_{N} &= -L \ket{0}\bra{0} -R\ket{N}\bra{N} + H_{N}^{\textrm{walk}}, \label{HwalkLR}
\end{align}
where $H_{N}^{\textrm{walk}}$ is minus the adjacency matrix of a line with $N$ links \eqref{Hwalk}.
Using this notation, the Laplacian walk \eqref{HLapends} can also be written as
\begin{align}
	H_{N}^{\textrm{L}} = 2\ii + H^{(1,1)}_{N}. \label{HLapLR}
\end{align}
The biased walk \eqref{biasedclock} can be seen as a special case of $H^{(L,R)}$ with $LR=1$ by setting $L=B$ and $R=\frac{1}{B}$, an extra prefactor and a constant shift, as seen from \eqref{rewriteHbias}.

Our analysis is similar to Childs' \cite{ChildsTrees}, where two identical quantum walks on a line are joined by an edge of different strength. There, the symmetric and antisymmetric sectors can be mapped to a quantum walk on a single line with one endpoint projector. A similar case is the topic of \cite{NagajKieferovaSomma}, where two connected walks on a line include endpoint projectors with strength $s$ and $1-s$.
Here, we talk about general endpoint projectors, in particular positive as well as negative $L,R$.

Because the geometry of the system is mostly a line, the eigenvectors there can be of only two types.
We call the first class {\em goniometric} -- these eigenvectors are combinations of plane waves:
\begin{align}
	\ket{g_p} &=  \sum_{x=0}^{N} \left( a e^{-ipx} + b  e^{ipx} \right)\ket{x},  \qquad
	\label{gonsol}
	E^{(g)}_p = -2\cos p,
\end{align}
with $0\leq p < 2\pi$. Note that $-2 \leq E_p^{(g)} \leq 2$.
We call the second class of eigenvectors {\em hyperbolic} -- they are combinations of hyperbolic functions, obtained by using imaginary momenta $p=iq$:
\begin{align}
	\ket{h_q} &= \sum_{x=0}^{N} \left(c e^{-qx} + d e^{qx} \right)\ket{x},  \qquad
	\label{hypsol}
	E^{(h)}_q = -2\cosh q,
\end{align}
with $q>0$. Note that $E_q^{(h)} <- 2$.

Our goal is to estimate the gap of \eqref{HwalkLR}. We will investigate the outlying energies for the goniometric solutions, as well as the existence of the hyperbolic solutions.
First, let us look at the goniometric solutions. At the endpoints of the line ($x=0$ and $x=N$), the eigenvalue equation $H_N^{(L,R)}\ket{g_p} = E^{(g)}_p \ket{g_p}$ reads\footnote{A quick way to arrive at the equation at the left endpoint is to realize that the contribution from the self-loop has to be the same as if it came from a line continued to $x=-1$, i.e. from a point with amplitude $a e^{ip} +b e^{-ip}$.}
\begin{align}
	- L\left( a + b \right) - \left(a e^{-ip} +b e^{ip}\right)
		&= -\left(e^{ip} +e^{-ip}\right) (a+b),   \label{quantG1}\\
	-R\left( a e^{-iNp} + b e^{iNp} \right) - \left( a e^{-i(N-1)p} + b e^{i(N-1)p} \right)
			&= -\left(e^{ip} +e^{-ip}\right)\left( a e^{-iNp} + b e^{iNp} \right). \nonumber
\end{align}
We can rearrange these to get
\begin{align}
	a \left(L-e^{ip}\right) &= b \left(e^{-ip}-L\right), \label{quantG2}\\
	a \left(e^{-ip}-R\right) & = b\, e^{2iNp} \left( R - e^{ip}\right).  \nonumber
\end{align}
Let us first deal with special cases, making sure we don't divide or multiply by zero when simplifying the above equations. The complex numbers in the brackets in \eqref{quantG2}
are real only for $p \in \{0,\pi\}$, for which $e^{ip}=e^{-ip}$, so only the $a$ part in \eqref{gonsol} has a distinct meaning and we can just set $b=0$. The option $p=0$ then gives a nonzero $a$ if and only if $L=R=1$, when we get a a special solution: the uniform superposition
\begin{align}
	\ket{g^{(L=R=1)}_0} = \frac{1}{\sqrt{N+1}}\sum_{x=0}^{N} \ket{x}
	\label{LR1solution}
\end{align}
with energy $E_0 = -2$.
Similarly, a solution for $p=\pi$ with energy $E_\pi = 2$ exists if and only if $L=R=-1$:
\begin{align}
	\ket{g^{(L=R=-1)}_\pi} = \frac{1}{\sqrt{N+1}}\sum_{x=0}^{N} (-1)^x \ket{x}.
\end{align}
We can also verify that $a=0$ (or $b=0$) doesn't work except in the above special cases.
Therefore, we can multiply the equations in \eqref{quantG2} together, get rid of $ab$ and obtain a quantization condition for the momentum $p$:
\begin{align}
e^{i2Np}
	\left( R -e^{ip}\right)
			\left(L-e^{ip}\right)
	&= \left(R-e^{-ip}\right)
			\left(L-e^{-ip}\right).
	\label{quantG}
\end{align}
This condition has the form $e^{i2Np} = \frac{v}{v^*}\frac{w}{w^*}$ with $v=R-e^{-ip}$ and $w = L-e^{-ip}$. It ties together the arguments of the complex numbers via
\begin{align}
	2Np + 2\pi k = 2\textrm{arg}(v)\,+2\,\textrm{arg}(w)
	= 2\arctan \frac{\sin p}{R-\cos p	} + 2\arctan \frac{\sin p}{L-\cos p	}. \label{vwarg}
\end{align}

We can easily calculate the arguments in some special cases, useful for the proof of a lower bound on the promise gap of Kitaev's Local Hamiltonian in Section 3. Setting $R=1$ gives us $2\textrm{arg}(v)= -p+\pi$, while $R=0$ means $2\textrm{arg}(v)=-2p$, and $R=-1$ results in $2\textrm{arg}(v)= -p$. These special points are of interest for the calculation in Section~\ref{sec:KitaevGap}.

Recall that we are interested in the eigenvectors and eigenvalues near the bottom of the spectrum.
To investigate the lowest possible values of $E_p = -2 \cos p$, we thus need to look at $p\rightarrow 0$.
There, assuming large $N$ and $R,L$ bounded away from 1, we can expand the arguments in \eqref{vwarg}
using the Taylor series and obtain solutions with energies near the bottom and top of the spectrum, at points
\begin{align}
	p_k = \frac{k\pi}{N-\frac{1}{R-1}-\frac{1}{L-1}} + o\left(N^{-1}\right),
\end{align}
for small integers $0<k\ll N$.
The lowest (nonzero) magnitude $p$ for a goniometric state is thus $\oo{N^{-1}}$.
It means this state has energy at least $\oo{N^{-2}}$ higher than the $p=0$ state (if it exists) or any hyperbolic solution (if it exists) that we find below.

Second, let us analyze the hyperbolic solutions. Analogously to \eqref{quantG1} and \eqref{quantG2}, we now get
\begin{align}
	c\left( L -e^{q}\right) &= d \left(e^{-q}-L\right),  \label{quantHstart}\\
	c \left(e^{-q}-R\right) & = d\, e^{2Nq} \left( R -e^{q}\right) . \nonumber
\end{align}
Again, let us first check for special cases.
Picking $c=0$ implies $e^q = \frac{1}{L} = R$, which only works in the special case $LR=1$ discussed in more detail below in Section~\ref{sec:LR}. Such an eigenvector has amplitudes falling off exponentially when moving away from one of the ends. Picking $d=0$ is just like $c=0$ but with exchanging $L\leftrightarrow R$ or $q \leftrightarrow -q$. Finally, when one of the terms in brackets is zero, it implies $c=0$ or $d=0$ and reverts to the above.
With this in mind, we are free to multiply the equations together, and obtain a quantization condition for $q$:
\begin{align}
	e^{2Nq} \left(R-e^{q}\right)\left(L-e^{q}\right)
	&=
	\left(R-e^{-q}\right)
	\left(L- e^{-q}\right).
	\label{quantH}
\end{align}
Let us analyze the behavior of this equation for large $N$.

Without loss of generality, we can assume $L\leq R$ and $q>0$ ($e^q>1$). Note that choosing $q=0$ produces the same state \eqref{LR1solution} as $p=0$ discussed above, while $q<0$ just exchanges $c$ and $d$ in \eqref{hypsol}. Let us label $y=e^q$, observe that $q>0$ implies $y>1$, and rewrite \eqref{quantH} as
\begin{align}
	y^{2N} (R-y)(L-y) = \left(R-\frac{1}{y}\right)\left(L-\frac{1}{y}\right).
	\label{quant2}
\end{align}
What happens when we start near $y=1$ ($q=0$) and start increasing $y$?
Exactly at $y=1$, the two sides of the equation are equal to $(R-1)(L-1)$.
Next, assuming large $N$, the growth of $y^{2N}$ with increasing $y$ dominates everything.
However, the terms in the brackets on the left can become very small near $y=R$ (and $y=L$). The left side changes sign at $y=R$ (or $y=L$), and again quickly reaches large magnitude. Thus, it must achieve the value of the right side (its magnitude is for constant $R\neq 0$ bounded from above by another constant) very close
to the point $y=R$ (and $y=L$). This is how hyperbolic eigenvectors appear, with corresponding eigenvalues
\begin{align}
	E_q = - 2\cosh q = - \left(y + \frac{1}{y}\right)
	\approx - \left(R +\frac{1}{R}\right) = -\frac{1+R^2}{R},
	\label{Rsol}
\end{align}
and of course, $-\frac{1+L^2}{L}$. Because $y>1$, this is possible only if $R>1$ (and similarly, $L>1$).

In the special case $R=L=C>1$, we get two solutions with energy near $-\frac{1+C^2}{C}$. 
There is an exponentially small (in $N$) energy split between them; the state with lower energy is symmetric and the other antisymmetric across the middle of the chain.
On the other hand, because we have $y>1$, when both $R,L\leq 1$, the left side of \eqref{quant2} is strictly larger even without the $y^{2N}$, i.e. $(R-y)(L-y)>(R-1/y)(L-1/y)$, and no hyperbolic eigenvectors exist.
Altogether, we have
\begin{align}
	\begin{array}{r|l}
	& \textrm{hyperbolic solutions} \\
	\hline
	R\geq L>1  & \textrm{two solutions near $e^q = R$ and $e^q = L$}, \\
	\hline
	R>1\geq L & \textrm{one solution near $e^q = R$}, \\
	\hline
	1\geq R\geq L & \textrm{no hyperbolic solutions}.
\end{array}
\end{align}

Therefore, there can be at most two hyperbolic solutions, with eigenvalues \eqref{Rsol} for $R>1$ and similarly for $L>1$. These are at least a constant below $-2$, and for $R\neq L$ also a constant away from each other. In that case, the Hamiltonian is gapped. If both $R,L\leq 1$, there are no hyperbolic solutions, and the gap of the Hamiltonian scales as $\oo{N^{-2}}$.

We can say more about the whole spectrum in the special cases solved below.

\subsubsection{The biased walk}
\label{sec:LR}

Up to a constant shift in energy, the biased walk of Section~\ref{sec:biasedwalk} is the special case of $H^{(L,R)}_N$ with $R = \frac{1}{L} = B > 1$. In this case we can solve the quantization conditions exactly.
Labeling $y=e^q$, the hyperbolic quantization condition \eqref{quantH} becomes
\begin{align}
	\frac{y^{2N}}{B} (B-y)\left(1-yB\right) = \frac{1}{y}\left(By-1\right)\frac{1}{By}\left(y-B\right),
\end{align}
which is {\em exactly} fulfilled only for $y=e^{q} = B$.
This means a single hyperbolic solution, falling off away from the right end, as choosing $c=0$ in \eqref{hypsol} is viable for $e^q=B$ (or $d=0$ and $-q$, which is the same thing, satisfying \eqref{quantH}. The energy of this state is $E = -2\cosh q = -B-\frac{1}{B} = -\frac{1+B^2}{B}$. Note that for $R=B$,
we can rewrite $H^{\textrm{bias}}_N$ \eqref{rewriteHbias} as $\left(1+B^2\right)\ii + BH^{\left(B,\frac{1}{B}\right)}_N$. The energy of this hyperbolic solution for \eqref{rewriteHbias} is thus exactly zero, $H^{\textrm{bias}}_N$ is frustration free and this eigenvector is its unique 0-energy ground state.

Furthermore, we can also express the quantization condition \eqref{quantG} for the goniometric solutions as
\begin{align}
	e^{i2Np}
	\left( B - e^{ip}\right)
			\frac{1}{B}\left(1-Be^{ip}\right)
	&= e^{-ip} \left(Be^{ip}-1\right)
		\frac{e^{-ip}}{B}  \left(e^{ip}-B\right).
\end{align}
We know that $p=0$ works only for $L=R=1$ and $p=\pi$ only for $L=R=-1$. Thus, we can get rid of the nonzero factors, and simplify the condition to $e^{i2(N+1)p}=1$, which means
\begin{align}
		p = \frac{k\pi}{N+1}, \quad k=1,\dots,N. \label{pcondition}
\end{align}
Note that for $B>1$, the $k=0$ solution doesn't exist. Altogether, we are getting $N+1$ solutions (1 hyperbolic, $N$ goniometric), as we should.
All the eigenvalues for $H^{\left(B,\frac{1}{B}\right)}_N$ are thus
\begin{align}
	E_0 = - \left(B+\frac{1}{B}\right), \qquad \textrm{and} \qquad
	E_k = - 2\cos \frac{k\pi}{N+1}, \quad k=1,\dots,N. \label{HwalkSpecialE}
\end{align}
The difference between the {\em highest} two eigenvalues of the $R=\frac{1}{L}=B$ special case of \eqref{HwalkLR} is thus $\oo{N^{-2}}$ (the difference between two cosines of nearby $k$ values).
What is more interesting, the gap of \eqref{HwalkLR}, the difference of its two lowest eigenvalues, is lower bounded by a constant for $B+\oOmega{1}>1$, as the only hyperbolic solution energy $-\left(B+\frac{1}{B}\right)$ is bounded away from $-2$, the lower bound on the energies of the goniometric solutions.

As we promised to show in the beginning of Section~\ref{sec:biasedwalk}, this translates easily to a constant gap for $H^{\textrm{bias}}_N$, as it is just a rescaled and shifted version of $H^{\left(B,\frac{1}{B}\right)}_N$.

We can also use the above result to investigate the spectrum of the Hamiltonian $H^{\left(2-2s,2s\right)}_N$ for $s$ close to $\frac{1}{2}$, i.e. the Hamiltonian $H^{\left(1-x,1+x\right)}_N$ for small $x=\frac{1}{2}\left(s-\frac{1}{2}\right)$. This will be useful when planning to use these Hamiltonians for adiabatic quantum computation.
The equation $LR=1$ is approximately fulfilled with error $\oo{x^2}$, and the argument leading to
\eqref{pcondition} is valid up to error $\oo{x^2}$ in $p$. This translates to eigenvalues
\begin{align}
	E_0 = - 2 - x^2 + \oo{x^3}, \qquad \textrm{and} \qquad
	E_k = - 2\cos \frac{k\pi}{N+1} + \oo{\frac{x^2}{N}+x^4}, \quad \textrm{for }k\ll N, \label{HwalkS1MS}
\end{align}
and a spectral gap with an x-dependent lower bound
\begin{align}
	\Delta \geq x^2 + \frac{\pi^2}{(N+1)^2} + \oo{\frac{x^2}{N}+x^3}.
\end{align}
This will be useful later in Section \ref{sec:adiabatic}, when we will employ Hamiltonians of the form
$H_{N}^{(2-2s,2s)}$ for adiabatic quantum computation.

\subsubsection{Unit-strength or no self-loops at the ends}
\label{sec:walk11}

Finally, let us finish with the analysis of a few special values for $L$ and $R$.

First, we look at the chain \eqref{HwalkLR} with $L = \frac{1}{R} = R = 1$, i.e. the Hamiltonian $H^{(1,1)}_{N}$.
As noted above, it has no hyperbolic (exponentially growing) solution.
However, the point $p=0$ in \eqref{pcondition} is now also available, producing the uniform superposition state. The simplest expression for all $N+1$ eigenvalues is then
\begin{align}
	E_k = - 2\cos \frac{k\pi}{N+1}, \quad k=0,\dots, N. \label{E11}
\end{align}
The smallest one is $E_0=-2$. The separation of momenta is at least $\frac{\pi}{N+1}$, so the low-lying as well as high-lying gap (top of the spectrum) of the Hamiltonian
$H^{(1,1)}_{N}$, as well as of the rescaled and shifted Hamiltonian \eqref{HLapends} is again $\oo{N^{-2}}$.
In \eqref{HLapLR}, we have seen that the Laplacian walk Hamiltonian $H_{N}^{\textrm{L}}$ is a simple shift of $H^{(1,1)}_{N}$. Thus, $H_{N}^{\textrm{L}}$ has ground state energy $E_0^{\textrm{L}}=0$ and a gap on the order of $\oo{N^{-2}}$.

Second, we can similarly analyze the Hamiltonian $H^{(-1,-1)}_{N}$.
It has no hyperbolic solutions. The momentum $p=\pi$ is a solution of \eqref{quantG},
and produces the uniformly alternating state. The simplest expression for all $N+1$ eigenvalues is then again \eqref{E11}, but this time for $k=1,\dots,N+1$.
The smallest one is $E_1=-2 + \oo{N^{-2}}$, and the largest is $E_{N+1}=2$, while the gap at the bottom and top of the spectrum has again size $\oo{N^{-2}}$.

Third, let us look at $H^{(1,0)}_N$. This case is important for the proof in Section~\ref{sec:KitaevGap}. Because $L,R \leq 1$, it has no hyperbolic solutions. Setting $R=0$ and $L=1$ and following the argument below \eqref{quantG} results in
$e^{i2Np} = -e^{-ip}e^{-2ip} $, meaning the solutions are
\begin{align}
	p = \frac{\pi(2k+1)}{2N+3}, \qquad k=0,\dots,N. \label{10sols}
\end{align}
The lowest possible $p$ is thus $p_0 = \frac{\pi}{2N+3}$,
the ground state energy is $-2+\oo{N^{-2}}$,
and the spacing in $p$ is $\oo{N^{-1}}$ so the gap is again $\oo{N^{-2}}$.
Next, plugging $p_0$ into \eqref{quantG2} results in $b=-ae^{2\pi i/(N+3)}$, which lets us compute the form of the ground state itself. Here we just present two observations. The normalization is $|a|=\oo{N^{-1/2}}$, the left end (where there is a self-loop) has amplitude
of magnitude $\oo{N^{-1/2}}$, and the right end (without a self-loop) has amplitude $\oo{N^{-3/2}}$.

Fourth, we finally look at $H^{(0,0)}_N$, again useful for the proof in Section~\ref{sec:KitaevGap}. Because $L,R<1$, it has no hyperbolic solutions. Setting $R=0$ and $L=0$ in \eqref{quantG}, using the argument just below that equation results in
$e^{i2Np} = e^{-4ip} $, meaning the solutions are
\begin{align}
	p = \frac{\pi k}{N+2}, \qquad k=1,\dots,N+1. \label{00sols}
\end{align}
The lowest possible $p$ is thus $p_1 = \frac{\pi}{N+2}$,
the ground state energy is $-2+\oo{N^{-2}}$, the spacing in $p$ is $\oo{N^{-1}}$ so the gap is again $\oo{N^{-2}}$.

\section{A new promise gap bound for Kitaev's QMA-complete Local Hamiltonian}
\label{sec:KitaevGap}

Recall from Section~\ref{sec:LH}, that Kitaev's Local Hamiltonian problem is QMA complete, if there is a promise that the lowest eigenvalue is below $E_a$ or above $E_b$, with a promise gap $E_b-E_a = \oOmega{N^{-3}}$. Here, we prove that the problem remains QMA complete even if the promise gap is $E_b-E_a = \oOmega{N^{-2}}$.
After finishing this proof, we have learned that this result has been also independently proved by Bausch and Crosson \cite{BauschCrossonGap},
by Markov chain mixing techniques. Our approach is different, as we rely on quantum walks and the results on their spectra that we have derived in Section~\ref{sec:biasedwalk}.

Kitaev \cite{KitaevBook} used a general geometric lemma for a sum of two positive semidefinite operators to bound the lowest eigenvalue of the Hamiltonian \eqref{HKitaev}. However, now that we understand biased quantum walks, we can calculate bounds on the eigenvalues directly, getting an improved promise bound gap.

\begin{theorem}[Kitaev's 3-local Hamiltonian problem with an improved promise bound]
\label{th:gap}
It is QMA-complete to determine whether the ground state energy of a Hamiltonian with $\poly(N)$ constant-norm, 3-local terms for an $N$-qubit system is $\geq E_1$ or $\leq E_0$ for $E_1-E_0 = \oOmega{N^{-2}}$.
\end{theorem}

Our proof below is based on the original 5-local Hamiltonian construction \cite{KitaevBook}. However, it also works for other constructions with {\em no inherent bad clock transitions} in the Hamiltonian that need to be energetically punished. In particular, it works without change for the new 3-local Hamiltonian construction of Mozes \& Nagaj \cite{new3local}, resulting in Theorem~\ref{th:gap} stated above.
Note though, that it does not work for the constructions \cite{KR03}, \cite{KKR06}, \cite{OliveiraTerhal} or \cite{8state}, which all include bad clock transitions that need to be dealt with by a clock-checking Hamiltonian and a projection lemma that then implies a smaller eigenvalue as well as promise gap.

The proof that the 3-local Hamiltonian with a $\oOmega{N^{-2}}$ promise gap is still QMA-complete is a straightforward plugin (just a stronger analysis) within Kitaev's proof of QMA-hardness of 5-LH. There, a Hamiltonian is chosen so that its ground state is related to the history state of a quantum verification circuit with $N$ gates that has completeness $1-\epsilon$ and soudness $\epsilon$. In our proof, we will require a small $\epsilon=\oOh{N^{-2}}$, easily achievable by amplification. When starting with a verifier circuit with constant soudness, we can obtain an amplified verifier circuit with $\epsilon = \oOh{N^{-\alpha}}$ by using at most $\oOh{\log N}$ copies of the circuit \cite{KitaevBook}, or by a same width circuit with $\oOh{N\log N}$ gates of the witness-reusing alternating-measurement method by Marriott-Watrous \cite{MarriottWatrous} or amplification by phase estimation \cite{fastQMAamp}.

What is new in our proof is the lower bound on the ground state energy of $H$ in the {\em no} instances, i.e. when the original circuit accepts no state with probability more than $\epsilon$. Later, we will also sketch how to get an upper bound on the ground state energy when there exists a witness accepted by $U$ with probability at least $1-\epsilon$. Together, they will mean a relaxation on the conditions on the promise gap for the Local Hamiltonian problem.

\begin{proof}
Let us then look at the {\em no} case.
The Hamiltonian \eqref{HKitaev} is a sum of four terms: $H_{\textrm{prop}} + H_{\textrm{clock}} + H_{\textrm{init}} + H_{\textrm{out}}$. We choose to look at the standard 5-local implementation of the domain-wall clock with unique forward and backward clock transitions. Therefore, the Hilbert space
splits into the invariant {\em good} subspace with proper single domain wall clock states,
and another invariant subspace with bad clock states. The states in the bad clock subspace have at least one bad domain wall ``$01$'' in the clock register, and thus have energy at least a constant, coming from the term $H_{\textrm{clock}}$. However, we know that the ground state energy of an ansatz (any history state with proper ancilla initialization) is lower than $\oo{N^{-1}}$, because it is ``detected'' only by the readout term, and that part of the state has amplitude $\oo{N^{-1/2}}$
Thus, the actual ground state must come from a state in the {\em good} subspace.

If we had no endpoint projectors (checking the ancilla initialization and the readout), the structure of the good subspace would be simple, given by how $H_{\textrm{prop}}$ connects states to each other -- as quantum Laplacian walks on a line of the type \eqref{HLap}. However, with the terms $H_{\textrm{init}}$ and $H_{\textrm{out}}$, we need to work a bit harder.
First, we will find a convenient basis for the Hilbert space of the data register, thanks to Jordan's lemma about two projectors. Second, we will append the clock register, find a basis for the whole good subspace,
and show how our Hamiltonian has a simple form. Third, we will analyze this simplified Hamiltonian
and prove a lower bound on the ground state energy in the {\em no} case (and an upper bound in the {\em yes} case) of the original QMA-complete problem instance.

The initialization term $H_{\textrm{init}}$ in Kitaev's Hamiltonian \eqref{HKitaev} is a sum of projectors on the $\ket{1}$ states of the ancillas in the beginning of the computation, under which the states with more badly initialized ancillas have higher energies. We can only decrease the ground state energy if we instead choose an initialization term that is itself a projector
\begin{align}
	H_{\textrm{init}}' &= \ket{0}\bra{0}_{\textrm{clock}} \otimes P_{\textrm{data}}, \label{Hinitmodified}\\
	P &= \ii - \ket{0\cdots 0}\bra{0\cdots 0}_{\textrm{ancillas}} \nonumber,
\end{align}
where $P$ is a projector acting on the data register only, giving energy $1$ to all states that do not have all ancillas $\ket{0}$ at the start of the computation. We choose to use this modified version because dealing with the projector $H_{\textrm{init}}'$ is simpler than dealing with the positive semidefinite (sum of projectors) $H_{\textrm{init}}$.

Let us consider another projector $Q$ acting on the data register, related to the term $H_{\textrm{out}}$ in \eqref{HKitaev}:
\begin{align}
	H_{\textrm{out}} & = \ket{N}\bra{N}_{\textrm{clock}} \otimes \left(\ii \otimes \ket{0}\bra{0}_{\textrm{out}}\right)_\textrm{data} =  \ket{N}\bra{N}_{\textrm{clock}} \otimes\Pi_{\textrm{out}}^{0}, \label{HoutfromQ}\\
	Q &= U^\dagger \Pi_{\textrm{out}}^{0} U.
\end{align}
The projector $P$ keeps states with nonzero ancillas intact, while the projector $Q$ projects on non-accepted states. These two projectors give a lot of structure to the Hilbert space of the data register. Furthermore, they let us investigate invariant subspaces of the whole Hilbert space of the clock and data registers, and the form of the Hamiltonian there will be amenable to analysis.

\begin{figure}%
\begin{center}
\includegraphics[width=14cm]{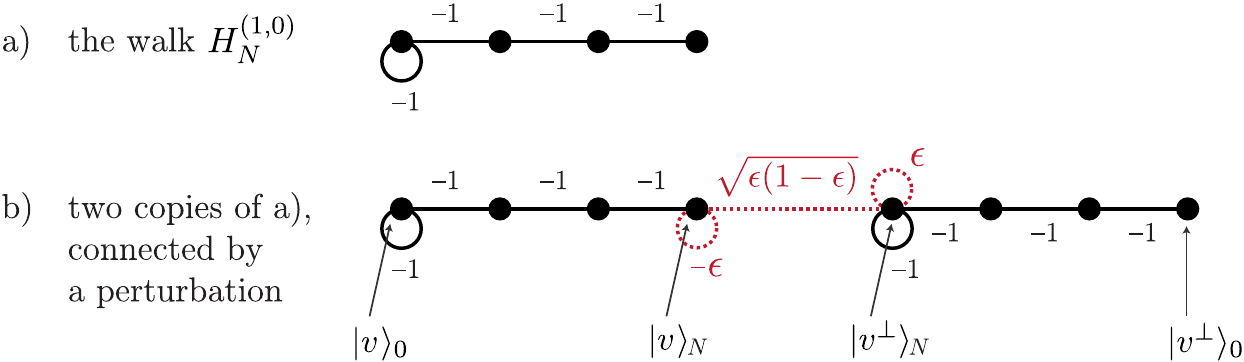}%
\caption{The special cases of walks with endpoint projectors appearing in the proof of the promise gap lower bound. a) The walk $H^{(1,0)}_N$ appears in cases 2,3 of the 1D invariant subspaces.
b) In the 2D invariant subspaces, we get two copies of the walk $H^{(1,0)}_N$, connected by a perturbation \eqref{perturbationsigma}.
}%
\label{fig:H10walk}%
\end{center}
\end{figure}

According to Jordan's lemma (see e.g. \cite{fastQMAamp}) for a pair of projectors $P,Q$, a Hilbert space can be split into 1D and 2D subspaces invariant under $P$ and $Q$.
First, let us look at the simpler, 1D subspaces, where the vectors are simultaneous eigenvectors of $P$ and $Q$.
Because $P,Q$ have eigenvalues $0,1$, there are four possibilities:
\begin{enumerate}
	\item The vector $\ket{u}$ obeys $P\ket{u} = Q\ket{u} = 0$. In the {\em no} instance of the problem, there is no such state, both properly initialized and fully accepted.
	\item The vector $\ket{u}$ obeys $P\ket{u} = 0$ and $Q\ket{u} = \ket{u}$. It is a state with proper ancillas, that is fully rejected by the circuit $U$. Such states might exist. Let us take this state $\ket{u}$ of the data register, append the clock register, and look at the subspace spanned by the basis
	\begin{align}
		\ket{u}_t = \ket{t}_{\textrm{clock}} \otimes U_t \dots U_1 \ket{u}_{\textrm{data}},
		\qquad t \in \{0,\dots,N\},\label{utbasis}
	\end{align}
where the original verifier circuit can be decomposed as $U = U_t \dots U_1$.
It is clear that this subspace is invariant under the terms $H_{\textrm{clock}}$, $H_{\textrm{prop}}$ and $H_{\textrm{init}}'$ of Kitaev's Hamiltonian. Note that $H_{\textrm{out}}$ also leaves this subspace invariant, as
\begin{align}
	H_{\textrm{out}} \ket{u}_N &= H_{\textrm{out}} \ket{N} \otimes U \ket{u}
	= \ket{N} \otimes \Pi_{\textrm{out}}^{0} U \ket{u}
	= \ket{N} \otimes U U^\dagger \Pi_{\textrm{out}}^{0} U \ket{u} \nonumber\\
	&= \ket{N} \otimes U Q \ket{u} = \ket{N} \otimes U \ket{u} = \ket{u}_N,
\end{align}
because the vector $\ket{u}$ obeys $Q\ket{u} = \ket{u}$. Thus, we have an invariant subspace spanned by the basis \eqref{utbasis}. Let us look at what form the Hamiltonian gets on the line of states \eqref{utbasis}. We have $H_{\textrm{clock}}\ket{u}_t=0$, as we are here talking only about proper clock states. The term $H_{\textrm{prop}}$ becomes a Laplacian-type walk \eqref{HLapends} on a line of the states $\ket{u}_t$. Next, we have $H_{\textrm{init}} \ket{u}_1 = 0$, because $P\ket{u}=0$. Finally, because the state $\ket{u}$ is fully rejected by the verifier, $H_{\textrm{out}}\ket{u}_N = \ket{u}_N$ translates to an extra right endpoint projector on the line. Altogether, in the basis \eqref{utbasis} we get the Hamiltonian $H_N^{\textrm{L}} + \ket{N}\bra{N} = 2\ii + H_{N}^{(1,0)}$ depicted in Figure~\ref{fig:H10walk}a.
Using the ground state energy of $H_{N}^{(1,0)}$ from Section~\ref{sec:walk11} below \eqref{10sols}, we get a lower bound on the energy of the shifted walk: $E \geq 2 - 2\cos\frac{\pi}{2N+3}= \frac{\pi^2}{4N^2} - \oOh{N^{-3}} =  \Omega(N^{-2})$.

	\item The vector $\ket{u}$ has bad ancillas and is fully accepted, i.e. $P\ket{u}=\ket{u}$ and $Q\ket{u}=0$. It again defines an invariant subspace of the data+clock registers spanned by the basis \eqref{utbasis} (a line of states). Similarly to the previous case, on this line, the Hamiltonian becomes a Laplacian walk with an extra endpoint projector on the left end: $H_N^{\textrm{L}} + \ket{0}\bra{0}$, as we now have $H_{\textrm{init}}'\ket{u}_0 = \ket{0}\otimes P\ket{u} = \ket{u}_0$, and $H_{\textrm{out}}'\ket{u}_N = 0$.
	The ground state energy of this Hamiltonian $2\ii + H_{N}^{(0,1)}$ is the same as in case 2, lower bounded by $\Omega(N^{-2})$.

	\item The vector $\ket{u}$ has nonzero ancillas and is also fully rejected, i.e. $P\ket{u}=\ket{u}$ and $Q\ket{u}=\ket{u}$.
	Yet again, it defines an invariant subspace of the whole Hilbert space of the clock and data registers, spanned by the basis \eqref{utbasis}.
	In this basis, the Hamiltonian is this time a Laplacian-type walk on a line with extra added endpoint projectors at both ends. Adding $\ket{0}\bra{0}$ and $\ket{N}\bra{N}$ to \eqref{HLapLR} gives us $2\ii + H_{N}^{(0,0)}$. Using the ground state energy of $H_{N}^{(0,0)}$ from Section~\ref{sec:walk11} below \eqref{00sols},
	we find a lower bound on the energy here: $E\geq 2 - 2\cos \frac{\pi}{N+2}
	= \frac{\pi^2}{N^2} - \oOh{N^{-3}}$, which is again $\Omega(N^{-2})$.
\end{enumerate}
We conclude that case 1 can't happen, while for large enough $N$, cases 2-4 give us a lower bound on the ground state energy $E\geq \frac{5}{2N^2}= \Omega(N^{-2})$.

Let us now deal with the 2D invariant subspaces that exist because the projectors $P$ and $Q$ are not orthogonal. For each such subspace $\mathcal{H}$, we can write down a basis $\{\ket{v},\ket{v^\perp}\}$ made from eigenvectors of $P$, with $P\ket{v} = 0$ and $P\ket{v^\perp} = \ket{v^\perp}$.
Note that
\begin{align}
	\bra{v} Q\ket{v}
	&= \bra{v} U^\dagger \Pi^0_{\textrm{out}} U \ket{v}
  = 1-p_v, \label{pv}
\end{align}
where $p_v\leq \epsilon$ is the acceptance probability of the original circuit $U$ for the state $\ket{v}$ (recall that it has properly initialized ancillas).
Because Jordan's lemma ensures the subspace is invariant under $P,Q$, the unnormalized states
$Q \ket{v}$ and $Q\ket{v^\perp}$ must also belong to the subspace $\mathcal{H}$.
In particular, using \eqref{pv} we can write
\begin{align}
	Q\ket{v} &= \left(1-p_v\right) \ket{v} + a \ket{v^\perp}, \\
	|a|^2 &= \bra{v}Q\ket{v} - 2\left(1-p_v\right)\bra{v}Q\ket{v} + \left(1-p_v\right)^2 \braket{v}{v}
	= p_v\left(1-p_v\right),
\end{align}
and choose the phase of $\ket{v^\perp}$ so that $a=\sqrt{p_v\left(1-p_v\right)}$ is real. We thus get
\begin{align}
	Q\ket{v} &= \left(1-p_v\right) \ket{v} + \sqrt{p_v\left(1-p_v\right)} \ket{v^\perp}, \\
	Q \ket{v^\perp} &= Q \frac{1}{\sqrt{p_v\left(1-p_v\right)}} \left(Q\ket{v} - \left(1-p_v\right) \ket{v}\right)
	= \sqrt{\frac{p_v}{1-p_v}}Q\ket{v} \\
	&= \sqrt{p_v\left(1-p_v\right)} \ket{v} + p_v \ket{v^\perp}.
\end{align}
In the basis $\{\ket{v},\ket{v^\perp}\}$ the operator $Q = U^\dagger\Pi^0_{\textrm{out}}U$ thus reads
\begin{align}
	Q = \left[
	\begin{array}{cc}
		1-p_v & \sqrt{p_v\left(1-p_v\right)} \\
		\sqrt{p_v\left(1-p_v\right)} & p_v
	\end{array}
	\right]
	= \ket{v}\bra{v} - \sqrt{p_v}\, \sigma_{p_v}.
	\label{Qpert}
\end{align}
It has two terms. The first is the projector $\ket{v}\bra{v}$.
The second term is a $\sqrt{p_v}$ multiple of a Pauli matrix $\sigma = \sqrt{p_v} Z - \sqrt{1-p_v} X$. We will think of it as a small perturbation (we know that $p_v \leq \epsilon$ is small)
\begin{align}
	V = \sqrt{\epsilon}\,\sigma_{\epsilon} =
	\sqrt{\epsilon} \left(\sqrt{\epsilon}Z - \sqrt{1-\epsilon}X\right).
	\label{perturbationsigma}
\end{align}

Similarly to what we did for the 1D subspaces, let us now append a clock register to our states $\ket{v}$ and $\ket{v^\perp}$
and investigate the subspace spanned by the basis states $\ket{v}_t$ and $\ket{v^\perp}_t$, defined by \eqref{utbasis}.
The terms of Kitaev's Hamiltonian \eqref{HKitaev} take a nice form within this subspace.
First, the clock Hamiltonian is satisfied, and does not act here, as the clock register has proper states. Second, the propagation Hamiltonian creates a Laplacian type quantum walk $H^{\textrm{L}}_N = 2\ii + H^{(1,1)}_{N}$ on the states $\ket{v}_t$ for $t=0,\dots,N$, as well as another
walk $2\ii + H^{(1,1)}_{N}$ on the states $\ket{v^\perp}_t$ for $t=N,\dots,0$.
Third, our modified initialization term \eqref{Hinitmodified} adds a projector onto the state $\ket{v^\perp}_0= \ket{0}\otimes \ket{v^\perp}$ (this turns the second walk into $2\ii + H^{(1,0)}_{N}$).
Fourth, according to \eqref{Qpert}, the readout term adds a projector onto the state $\ket{v}_N= \ket{N}\otimes \ket{v}$ (this turns the first walk into $2\ii + H^{(1,0)}_{N}$), as well as a perturbation term $V$ \eqref{perturbationsigma}
on the states $\ket{v}_N$ and $\ket{v^\perp}_N$.
Altogether, in this subspace, up to an overall shift by $2\ii$, we get two weakly coupled quantum walks on a line, with extra endpoint projectors, depicted in Figure~\ref{fig:H10walk}b.

The perturbation $V$ has norm $\sqrt{\epsilon}$, so
elementary perturbation theory then tells us that the the ground state energy
of Kitaev's Hamiltonian in this subspace changes by at most $\oo{\sqrt{\epsilon}}$.
In practice, the situation is even better because of the form of $V$ as well as the unperturbed ground states. The actual decrease in energy is $- \sqrt{\epsilon} \max_{\ket{\phi}} \bra{\phi}\sigma_{\epsilon}\ket{\phi}$ for $\ket{\phi}$ a ground state of the unperturbed Hamiltonian made from two copies of $H^{(1,0)}_N$.
The ground state subspace of $H^{(1,0)}_N\oplus H^{(1,0)}_N$ is doubly degenerate,
Because $\sigma_\epsilon$ is close to the Pauli $X$ up to error $\sqrt{\epsilon}$, we can maximize $\bra{\phi}\sigma_{\epsilon}\ket{\phi}$ by choosing $\ket{\phi}$ so that the amplitudes at $\ket{v}_N$ and $\ket{v^\perp}_N$ are the same. In Section~\ref{sec:walk11} below \eqref{10sols} we proved that the amplitudes of the ground state at the ends of one $H^{(1,0)}_N$ chain are $\oo{N^{-3/2}}$ and $\oo{N^{-1/2}}$.
The normalized  combination of such eigenvectors on the two copies of the walk maximizing $\bra{\phi}\sigma_\epsilon\ket{\phi}$ then has amplitude $\oOh{N^{-3/2}}$ at the connected endpoints (see Figure~\ref{fig:H10walk}b). This gives us $\bra{\phi}\sigma_{\epsilon}\ket{\phi} = \oOh{N^{-3/2}}$.
The perturbation in energy is thus at most $-\oOh{\epsilon^{1/2}N^{-3/2}} + \oOh{\epsilon}$
It is then enough to get an $\Omega(N^{-2})$ lower bound on the ground state energy by choosing e.g. $\epsilon = N^{-2}$.
We can make this choice by amplifying the original circuit.

Finally, in the {\em yes} case, the history state for a good witness accepted with probability $\geq 1-\epsilon$ has energy at most $\frac{\epsilon}{N}$, for our choice $\epsilon = \frac{1}{N^2}$.
Altogether, the lowest eigenvalue in the {\em yes} and {\em no} cases are
\begin{align}
		E_{yes} \leq \frac{\epsilon}{N} \leq \frac{1}{N^3}, \qquad E_{no} \geq \frac{const.}{N^2}.
\end{align}
Thus for a circuit amplified to soundness at most $\epsilon = \oOh{N^{-2}}$ and completeness at least $1-\epsilon$, we obtain a new promise gap $E_{no}-E_{yes} = \oOmega{N^{-2}}$.
\end{proof}
Note that Bausch \& Crosson have independently found a comparable (tight $\Theta(N^{-2})$ promise gap) result in \cite{BauschCrossonGap}.

\subsection{Universal adiabatic computation with a Laplacian and endpoint projectors}
\label{sec:adiabatic}

We can also use Kitaev's Hamiltonian to adiabatically prepare the state $U\ket{0\cdots 0}$, i.e. to perform universal quantum computation \cite{UniversalAdiabaticCircuitHamiltonian, AdiabaticQCEquivStandardQC}.
Here, we present another such scheme and find how its required runtime scales with $N$, the number of gates in the circuit $U$. It is a straightforward application of what we have learned about lower bounds on the gaps of the Hamiltonians that involve a walk on a line and endpoint projectors.
The goal is to investigate whether our knowledge of the gap could help us speed up when far from the small-gap region, while focusing on going slowly when near it, and thus cut down the required runtime of the preparation procedure. However, it turns out that a straightforward ``local adiabatic evolution'' approach of \cite{RolandCerf}
does not help us here because of the particular way the gap closes down. The gap here doesn't grow with an $N$ prefactor when moving around its minimum, while it did so for the Grover problem that Roland and Cerf \cite{RolandCerf} were investigating.

\begin{figure}%
\begin{center}
\includegraphics[width=9cm]{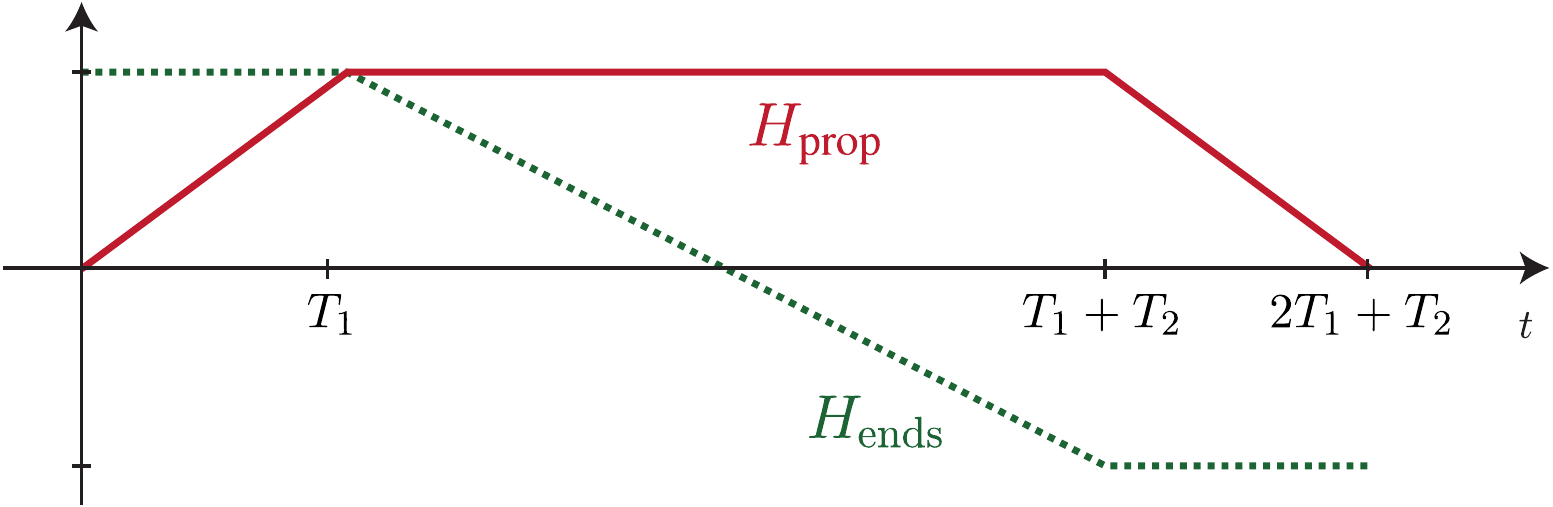}%
\caption{The schedule for universal quantum computation by adiabatic preparation using Kitaev's propagation Hamiltonian, and projectors on the initial and final states of the clock register. The gap is constant for the first and third sections, while in the middle section it becomes $\oo{N^{-2}+x^2}$ around $t = \frac{1}{2}(T_1+T_2) +x$.
}%
\label{fig:3wayadiabatic}%
\end{center}
\end{figure}

We assume $H_{\textrm{clock}} + H_{\textrm{init}}$ are always on, ensuring our playground (the low energy subspace) has proper clock register states, and that when the clock register reads $\ket{0}$, the data register is properly initialized to $\ket{0\cdots 0}$.
We propose a symmetric three-way adiabatic schedule, illustrated in Figure~\ref{fig:3wayadiabatic}, using Kitaev's propagation, clock and initialization terms \eqref{HKitaev}.
\begin{enumerate}
\item Start with
\begin{align}
	H(0) = H_{\textrm{clock}}+ H_{\textrm{init}} + H_{\textrm{ends}},
\end{align}
where $H_{\textrm{ends}}= - \ket{0}\bra{0}_{\textrm{clock}} + \ket{N}\bra{N}_{\textrm{clock}}$ prefers the clock state $\ket{0}$. There is a unique ground state of $H(0)$: the initial state $\ket{0}_\textrm{clock}\otimes \ket{0\cdots 0}_{\textrm{data}}$.
\item The first section takes constant time $T_1$. We turn on the propagation Hamiltonian as
\begin{align}
		H(t) = H_{\textrm{clock}} + H_{\textrm{init}} + H_{\textrm{ends}} + t H_{\textrm{prop}},
		\qquad 0\leq t \leq T_1.
\end{align}
\item The second section takes time $T_2$ scaling as $\oo{\epsilon^{-1}N^6}$. We flip the sign of $H_{\textrm{ends}}$
as
\begin{align}
		H(t) = H_{\textrm{clock}} + H_{\textrm{init}} + \left(1-2s(t)\right) H_{\textrm{ends}}
		+ H_{\textrm{prop}}, \qquad T_1\leq t \leq T_1+T_2,
\end{align}
with a monotonous parametrization $s(t)$, obeying $s(T_1)=0$ and $s(T_1+T_2)=1$.
\item Finally, again in constant time $T_1$, we turn off the propagation Hamiltonian as
\begin{align}
		H(t) = H_{\textrm{clock}} + H_{\textrm{init}} - H_{\textrm{ends}}+ (2T_1+T_2-t) H_{\textrm{prop}},
		\qquad T_1+T_2 \leq t \leq 2T_1+T_2.
\end{align}
\item We end with the Hamiltonian
\begin{align}
		H\left(2T_1+T_2\right) = H_{\textrm{clock}} + H_{\textrm{init}} - H_{\textrm{ends}},
\end{align}
at time $2T_1+T_2 = \oo{\epsilon^{-1}N^6}$. We claim the evolved initial state $\ket{0}_\textrm{clock}\otimes \ket{0\cdots 0}$ will be $\epsilon$ close to the desired state
$\ket{N}_\textrm{clock}\otimes U \ket{\phi}$.
\end{enumerate}

Let us analyze the gaps of these Hamiltonians, and then use the adiabatic theorem (Theorem 3 of \cite{JansenRuskaiSeiler}) to show that our preparation procedure works as promised.

Because of our choice of initial state, and having $H_{\textrm{clock}}$ and $H_{\textrm{init}}$ always on, the states that our Hamiltonians could possibly arrive at live within the subspace spanned by
\begin{align}
	\ket{\phi}_t = \ket{t}_{\textrm{clock}} \otimes U_t \dots U_1 \ket{0\cdots 0}. \label{abasis}
\end{align}
Let us express the Hamiltonians of our procedure in this basis. They are rather simple:
\begin{align}
	H_{\textrm{ends}} &= -\ket{\phi}_0\bra{\phi}_0 + \ket{\phi}_N\bra{\phi}_N,\\
	H_{\textrm{prop}} &= 2\ii + H_N^{(1,1)}.
\end{align}
In the first section, the relevant part of the Hamiltonian is $H_{\textrm{start}} + t H_{\textrm{prop}}$, a shifted and rescaled Laplacian walk with a varying left endpoint projector.
For $t=0$, the gap is $-2$, a constant, while for $t>0$ we have
\begin{align}
	H(t) = t \left( 2\ii + H_N^{(1+t^{-1},1-t^{-1})} \right). \label{Hsec1}
\end{align}
Recall from \eqref{Rsol}, that for $R<1$ and $L>1$, there is only one hyperbolic eigenstate of $2\ii+H_N^{(L,R)}$: the groundstate with energy near $2-L-1/L
= - (L-1)^2/L$. For \eqref{Hsec1}, it translates to
ground state energy
$-\frac{1}{t+1} \leq -\frac{1}{2}$.
Because the first excited state of \eqref{Hsec1} is a goniometric state, it has energy $\geq 0$.
Therefore, the gap of this Hamiltonian is at least a constant $\frac{1}{2}$.
The same holds for the third section in the schedule, where after the transformation $t' = 2T_1+T_2-t$, the Hamiltonian becomes
$H(t') = t' \left( 2\ii + H_N^{(1+t'^{-1},1+t'^{-1})} \right)$.

The middle section in Figure~\ref{fig:3wayadiabatic} is more interesting. There we have the Hamiltonian $2\ii + H^{(1,1)}_N - (1-2s) \ket{\phi}_0\bra{\phi}_0 + (1-2s) \ket{\phi}_N\bra{\phi}_N$, which can be rewritten in a simplified way as $2\ii + H^{(2-2s,2s)}_N$.
We have seen this Hamiltonian in Section~\ref{sec:LR}, and proven that it has a gap that is lower bounded by $\frac{\pi^2}{(N+1)^2}+ x^2 + \oo{x^2 N^{-1}}$, where $x=(2s-1)/4$.
Therefore, we can see the gap is smallest at $x=0$, i.e. at $s=1/2$, has magnitude $\oOmega{N^{-2}}$, and grows quadratically as $x^2$ when going away from this point.
We can thus straightforwardly apply Theorem 3 of \cite{JansenRuskaiSeiler}.
In this case, with a linear schedule, $T_2 = \oo{\epsilon^{-1}N^6}$
is surely enough for the final state to be $\epsilon$ close to the ground state of the final Hamiltonian. Following the local adiabatic evolution approach of \cite{RolandCerf}, which uses a specific slowdown that takes into account the gap dependence near $x=0$  could result in better scaling in $N$. However, this is not straightforward here. The calculations involve integrating the inverse cube or square of the gap (depending on the adiabatic theorem used). Because the gap dependence on $x$
in \cite{RolandCerf} was $\Delta(x) = \Delta_{\textrm{min}} + 2\Delta_{\textrm{min}}^{-1} x^2$, this has lead to improvement. This is not our case, as we have $\Delta(x) = \Delta_{\textrm{min}} + x^2$.

\section{Doing nothing (efficiently) can improve a computation}
\label{sec:idling}

In this Section, we will look at how to modify the unary clock construction
to achieve a high success probability for finding the computation done for Feynman's computer, and large overlaps of the ground state with the finished computation for Kitaev's Hamiltonian. Most importantly, our methods require only a few additional qubits (sublinear in the number of gates).

To compute with Feynman's computer, one needs to measure the clock register (pointer particle position). There is a chance to find it at the end of the computation, where all $N$ gates of the circuit have been performed. This is one of the model's drawbacks, as the probability of success is an inverse-polynomial in $N$.
One can choose a random time to measure, as Cesaro-mixing \cite{QWalksActaPhysica} guarantees the average time the computer spends in the final state is proportional to $\frac{1}{N}$ of the total time we run the computer.
Instead of running the computer for a randomly chosen reasonably long time, one can try to look at particular evolution times when the computation is more likely to be done, or to involve other tricks \cite{DeFalcoTamascelli}.

Another straightforward approach is to extend the quantum circuit with $N$ gates to $N+A$ gates, choosing to ``do nothing'' with the data for $A$ steps at the end of the computation. When we use gates $U_t = \ii$ for $N<t\leq N+A$, the fraction of states with the computation done becomes $\frac{1+A}{N+1+A}$. Moreover, for clock transitions with $t\geq N$, the data register remains unchanged, and so the required interactions involve only the clock register. Thus, all we need to do is to increase the number of possible clock states in the clock register.

One way to do this is to make the unary clock larger, adding $A$ clock qubits. In practice, this means we can tune the probability of finding the computation done as close to 1 as we want. What is the price we pay for this? First, the system size becomes $N+1+A$.
For example, we need $A=99 N$ extra clock qubits to guarantee a 99\% probability of success, which might be too costly. Second, the spectral as well as the promise gap of the Hamiltonian (see the calculations in the previous Section) system closes quadratically with the clock register size as $\oo{(N+A)^{-2}}$. This is bad news, if we want our computation to be resistant to noise.
Our answer to this problem are two solutions that require few additional clock qubits.
We believe these constructions will find their applications in universality as well as computational hardness results.

First, we present the {\em idling chain} construction in Section~\ref{sec:idling}, adding only a logarithmic number of clock qubits and their local interactions. It is designed for complexity applications, introducing a large overlap of the ground state with the finished computation, without a large increase in system size, without using large norm projectors, and without modifying the gap significantly. Note that it is not suitable for computation with Feynman's computer, because of reflection issues -- the time evolution of a computation does not smoothly transition between the unary clock and the idling clock.

The second, {\em cogwheel} construction in Section~\ref{sec:cog} is better suited for dynamical (Feynman computer) applications. It requires $2\sqrt{N+1+A}$ total clock qubits for all of the $N+1+A$ available clock states. The downside is the gap scaling as $\oo{(N+A)^{-2}}$, while clock qubit can be involved in up to $\sqrt{N+1+A}$, 3-local interaction terms. This model can be viewed as a combination of several wheels, where a full turn of one cogwheel advances the next wheel by one step.

\subsection{Idling the engine}

We now present a Hamiltonian whose ground state involves a uniform superposition over many clock states while leaving the data alone.
Our {\em idling chain} construction doesn't require a large number of extra clock qubits, while it can greatly increase the overlap of the ground state of Kitaev's Hamiltonian with the finished calculation (in the data register), while not decreasing the gap significantly. The detailed statement of the construction's properties is the following Theorem~\ref{th:idling}:
\begin{theorem}
\label{th:idling}
Consider the 4-local Hamiltonian for a unary clock with $N$ qubits connected to an idling chain with $C$ extra qubits and $C$ idling qubits. Let $A = 2^{C+1}-2$ and $z = (A+1)/N$.
\begin{enumerate}
\item It has a unique, zero-energy ground state $\ket{\psi}$: the uniform superposition of $N+1+A$ legal clock states, labeled by computational basis states $\ket{u}\ket{e}\ket{i}$ denoting the state of their unary, extra and idling qubits.
\item The expectation value of the projector $\Pi_{\textrm{done}} = \ket{1\cdots 1}\bra{1\cdots 1}\otimes \ii \otimes \ii$ in the ground state is
$\bra{\psi}\Pi_{\textrm{done}}\ket{\psi}=\frac{1+A}{N+1+A} = \frac{z}{z+1}$.
We thus call $\ket{\psi}$ the {\em amplified} history state.
\item The eigenvalue gap of the Hamiltonian is asymptotically lower bounded by
$\oOmega{N^{-2}}$ for $z=\poly(N)$.
\end{enumerate}
\end{theorem}

\begin{proof}
Let us describe the construction and prove its properties.
Recall from \eqref{history} that the {\em history state} of a quantum computation is
$
	\ket{\psi_{\textrm{hist}}} = \frac{1}{\sqrt{N+1}} \sum_{t=0}^{N} \ket{t}_{\textrm{clock}}\otimes \ket{\varphi_t}_{\textrm{data}}$, where
	$\ket{\varphi_t} = U_{t} U_{t-1} \dots U_{1} \ket{\varphi_0}$.
Instead of this, we want our ground state to be an {\em amplified} history state
\begin{align}
	\ket{\psi^{A}_{\textrm{hist}}} = \frac{1}{\sqrt{N+1+A}} \sum_{t=0}^{N+A}
		\ket{t}_{\textrm{clock}} \otimes \ket{\varphi_t}_{\textrm{data}},
		\qquad \textrm{with}\qquad
	\ket{\varphi_{t\geq N}} = \ket{\varphi_N}, \label{amphistory}
\end{align}
i.e. where we do nothing with the data register for clock steps $t\geq N$, leaving the computation done. Let us start with the original unary clock register with qubits $c_0, \dots, c_{N}$, and introduce $C$ {\em extra} unary clock qubits.
Below these, we add another row of $C$ {\em idling} qubits. 
\begin{align}
	\begin{array}{ccccc|cccc}
		c_1 & c_2 & c_3 & \cdots & c_{N+1} & c_{N+2} & c_{N+3} & \cdots & c_{N+1+C}  \\
		\hline
		& &  &  &  & i_{1} & i_{2} & \cdots & i_{C}
	\end{array}.
\end{align}
This clock register has $N+1+2C$ qubits. We want the {\em legal clock states} to have form
\begin{align}
	\begin{array}{cccccc|ccc}
		1 & \cdots & 1 & 0 & \cdots & 0 & 0 & \cdots & 0  \\
		\hline
		 & &  &  &  &  & 0 & \cdots & 0
	\end{array},
\quad \text{or} \quad
	\begin{array}{ccc|cccccc}
		 1 & \cdots & 1 & 1 & \cdots & 1 & 0  & \cdots & 0  \\
		\hline
		  &  &  & 0/1 & \cdots & 0/1 & 0  & \cdots & 0
	\end{array}.
	\label{idling2lines}
\end{align}
with the first row of qubits holding states with a single domain wall,
while the qubits in the {\em idling} row can take any value, as long as the corresponding domain-wall qubits above are in the state $\ket{1}$.

Let us build the Hamiltonian whose ground state is the uniform superposition of these legal states.
We start with the clock-checking Hamiltonian, which raises the energy of states other than with form \eqref{idling2lines}:
\begin{align}
	H_{\textrm{clock}} =
		\ket{0}\bra{0}_{c_1}
		+ \sum_{k=1}^{N+C} \ket{01}\bra{01}_{c_k c_{k+1}}
		+ \sum_{j=1}^{C} \ket{0}\bra{0}_{c_{N+1+j}} \otimes \ket{1}\bra{1}_{i_j}. \label{Hclockcheckidling}
\end{align}
Just as in \eqref{Hdwcc}, the first two terms ensure a single domain-wall (unary) signal in the top row qubits $c_0,\dots,c_{N+C}$. The last term in \eqref{Hclockcheckidling} allows the bottom row idling qubits to be on only if the unary qubit above them is also on.

Next, we need terms that energetically prefer a uniform superposition of these states.
First, we add the original 3-local domain-wall clock interactions
\begin{align}
		H_{\textrm{dw, L}} &=
						\sum_{j=1}^{N}
								\left(\ket{100}-\ket{110}\right)
								\left(\bra{100}-\bra{110}\right)_{c_{j}c_{j+1}c_{j+2}}. \label{HT1}
\end{align}
Then, we add 4-local terms\footnote{Note that if we also couple the data register to the 3-local original unary clock, adding interaction terms that involve 2-qubit gates, the interaction becomes naturally 5-local. However, we could keep it down to 3-local by e.g. using the construction from \cite{new3local}. On the other hand, the extra terms \eqref{Hi1} and \eqref{Hi2} do not involve the data, so the whole Hamiltonian including the data register can be made 4-local.} involving the additional qubits.
\begin{align}
		H_{\textrm{extra}} =
						\sum_{j=1}^{C-1}
								\ket{0}\bra{0}_{i_{j}} &\otimes \left(\ket{100}-\ket{110}\right)
								\left(\bra{100}-\bra{110}\right)_{c_{N+j}c_{N+j+1}c_{N+j+2}}
								 \label{Hi1} \\
							+ \ket{0}\bra{0}_{i_{C}}&\otimes \left(\ket{10}-\ket{11}\right)
								\left(\bra{10}-\bra{11}\right)_{c_{N+C}c_{N+1+C}},
						\nonumber
\end{align}
allows progress from the state $\ket{\cdots 100 \cdots}$ and return from
the state $\ket{\cdots 110 \cdots}$ only if the corresponding idling qubit is off.
Finally, we add the freewheeling term energetically preferring superpositions of idling qubit states if the extra qubits above are on:
\begin{align}
		H_{\textrm{idle}} &=
									\sum_{j=1}^{C}
										\ket{1}\bra{1}_{c_{N+1+j}}\otimes
										\left(\ket{1}-\ket{0}\right)
										\left(\bra{1}-\bra{0}\right)_{i_j}. \label{Hi2}
\end{align}
These terms together ensure that there is a unique, zero-energy ground state of $H_{\textrm{idling chain}} = H_{\textrm{clock}} + H_{\textrm{dw, L}} + H_{\textrm{extra}} + H_{\textrm{idle}}$: the uniform superposition of the legal states \eqref{idling2lines}.
This concludes proof of point 1 in Theorem~\ref{th:idling}.

\begin{figure}%
\begin{center}
\includegraphics[width=11cm]{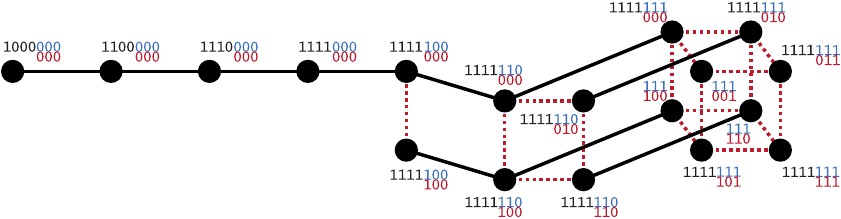}%
\caption{The legal clock states for an $N=4$ domain-wall clock with a $C=3$ idling chain.
A state is labeled by a string made from its $N$ unary bits (black),
extra $C$ unary bits (blue) and $C$ idling bits (red). Each line is a projector onto the antisymmetric combination of the clock states at the vertices.
The graph is a line (the original clock, first 4 vertices), connected to the idling part: a line connected to a square, to a cube, and so on.
	}%
\label{fig:idling43}%
\end{center}
\end{figure}

The {\em legal} clock state subspace $\mathcal{H}_{\textrm{legal}}$ spanned by the legal states \eqref{idling2lines} is invariant under $H_{\textrm{idling chain}}$.
In Figure~\ref{fig:idling43}, we illustrate this subspace and the possible transitions
\begin{align}
	\begin{array}{ccccccccccc}
		100 & \leftrightarrow & 110, & \qquad &
		100 & \leftrightarrow & 110, & \qquad &
		1 & \leftrightarrow & 1, \\
		&&& \qquad &
		0 & & 0 & \qquad &
		0 & & 1 \\
	\end{array}
\end{align}
given by \eqref{HT1}, \eqref{Hi1}, \eqref{Hi2}. There are $N+1$ states of the original unary clock and
\begin{align}
	A = 2 + 4 + \cdots + 2^C = 2\left(2^C-1\right)
\end{align}
extra clock states beyond $t=N$.
Altogether, there are $1+A$ states in which the computation can be considered done.
The ground state is a uniform superposition of legal states, i.e. the desired {\em amplified} history state \eqref{amphistory}.
It takes $2C$ extra qubits to achieve a finished computation ratio $\frac{1+A}{N+1+A}$, improving the original $\frac{1}{N+1}$. To achieve a constant ratio of states with the computation done compared to all states in the amplified history state, we can choose $A=\oo{N}$, for which we need $2C = \oo{\log_2 N}$ extra unary/idling qubits.
This results in point 2 of Theorem~\ref{th:idling}.

Concerning point 3 in Theorem~\ref{th:idling}, recall that the eigenvalue gap for the original domain-wall clock Hamiltonian is $\oo{N^{-2}}$. We will now prove a lower bound on the gap for our new the domain-wall clock with an attached idling chain.

Illegal clock states have constant and positive energy from the term \eqref{Hclockcheckidling}, so it is enough to find a lower bound on the gap in $\mathcal{H}_{\textrm{legal}}$. For this, we will first map the Hamiltonian $H|_{\textrm{legal}}$ to a stochastic matrix $P$ describing a random walk induced by transition rules. We will then relate the spectral gap $\Delta(H|_{\textrm{legal}})$ to the gap $1-\Lambda_1(P)$ by a similarity transformation. Finally, we will find a lower bound on $1-\Lambda_1(P)$ using canonical paths.

We can map our Hamiltonian to a random walk induced by the transition rules \cite{BravyiTerhalStoqFF}.
We define the matrix
\begin{align}
	P=\mathbb{I}-\frac{1}{2(C+1)}H|_{\textrm{legal}},\label{IdlingChainP}
\end{align}
acting on legal clock states.
Let us prove that it is a stochastic matrix, has a unique stationary distribution, and forms a reversible Markov chain with some nice properties.
Let us define $\pi(s)=\left| \langle s\ket{\psi}\right|^2 = (N+1+A)^{-1}$.
Because the uniform superposition over legal clock states $\ket{\psi} = \sum_{s} \sqrt{\pi(s)}\ket{s}$ is a zero eigenvector of $H|_{\textrm{legal}}$,
\begin{align}
	\sum_{t}  P_{s,t}
		= 1-\frac{1}{2(C+1)} \langle s |(H|_{\textrm{legal}}) \sum_t|t\rangle = 1 - 0 = 1.
\end{align}
Therefore, $P$ is a stochastic matrix. Next,
$\sum_s \pi(s)P_{s,t}=\pi(t)-\sum_s\pi(s) \frac{\langle s |(H|_{\textrm{legal}})|t\rangle}{2(C+1)}=\pi(t)$
means the uniform distribution $\pi(s)$ is a unique stationary distribution for $P$.
Since $H|_{\textrm{legal}}$ is real and symmetric, $P$ is also reversible: $\pi(s)P_{s,t}=\pi(t)P_{t,s}$.
Finally, $P_{s,t}=(2(C+1))^{-1}$ if states $s \neq t$ are connected by the Hamiltonian's transition rules,
and $P(s,s) \geq \frac{1}{2}$, as any state $s$ is involved in at most $(C+1)$ terms.

The spectral Gap of $H|_{\textrm{legal}}$ is related to the gap $\Delta_P = 1-\Lambda_2(P)$ of the Markov chain $P$ by 
\begin{align}
	\Delta_{H|_{\textrm{legal}}}=2(C+1) \Delta_P.
	\label{IdlingChainHtoP}
\end{align}
Thus, to find a lower bound on $\Delta_{H|_{\textrm{legal}}}$, we need to find an upper bound on $\Lambda_2(P)$, the second largest eigenvalue of $P$ (the largest eigenvalue of $P$ is 1 and corresponds to the stationary distribution $\pi(s)$).
This way we will find a bound on the gap of $P$.
For this, we will use the canonical path technique \cite{DiaconisStroock, Sinclair},
which says that for a family of canonical paths $\{\gamma_{s,t}\}$ connecting pairs of states $s,t$,
\begin{align}
	1-\Lambda_2(P) \geq (\rho l)^{-1}, \label{Pgap}
\end{align}
where $l=\max_{(s,t)}|\gamma_{s,t}|$ is the maximum length of a canonical path, and $\rho$ (the congestion) is defined by
\begin{align}
	\rho&=\max_{(a,b)\in E}\frac{1}{\pi(a)P_{a,b}}\sum_{(a,b)\in\gamma_{s,t}}\pi(s)\pi(t),\label{congestion}
\end{align}
where $(a,b)$ is an edge in the graph of $P$.
We will now construct a family of canonical paths between any pair of vertices $s,t$, and use its properties to bound $\Delta_P$.

Each vertex in Figure~\ref{fig:idling43}
is a computational basis state, so it can be labeled
by a string $d|e|i$, corresponding to the values of the original domain wall, extra, and idling qubits.
The vertices can be ordered by ordering the strings, e.g.
$1100|000|000 < 1111|110|010 < 1111|110|100 < 1111|111|001$.

\begin{definition}[Canonical paths]\label{CanonicalPaths} Take two states $s$ and $t$, and without loss of generality, assume $s<t$, meaning $d_s|e_s|i_s < d_t|e_t|i_t$ for their binary string labels.

In the special case $e_s=i_s=e_t=i_t = 0\cdots 0$ (both $s,t$ are not in the idling part yet), connect $s$ and $t$ through intermediate domain wall states of the form $1\cdots1 0\cdots 0|0\cdots 0|0\cdots 0$ as in \eqref{unarytrans}.

Next, when $e_s=i_s= 0\cdots 0$ but $e_t \neq 0\cdots 0$ ($s$ is from the original unary clock, but $t$ has nonzero extra qubits), connect the state $s$ to the state $s' = 1\cdots 1|0\cdots 0|0\cdots 0$ as above, and then continue by connecting $s'$ to $t$ as below.

To connect two states from the idling part, i.e. with the original domain wall qubits $d_s=d_t=1\cdots 1$, set $k=1$ and change the extra and idling strings one by one from left to right as follows:
\begin{enumerate}
\item If $(e_s)_k = (e_t)_k$, change $(i_s)_k$ to $(i_t)_k$ if necessary
as in \eqref{idlingtrans}. Increase $k$. Repeat this point until $(e_s)_k \neq (e_t)_k$, in which case continue to point 2, or until you reach the end of the idling chain ($k>C$).
\item We now have $(e_s)_k \neq (e_t)_k$. Change $(e_s)_k$ to $(e_t)_k$ as in \eqref{extratrans} and go back to point 1.
\end{enumerate}
\end{definition}

Let us find an edge that is included in the highest number of canonical paths. There are three types of edges. First, let us look at an edge on the unary part,
\begin{equation}
\begin{tikzpicture}[baseline=-.65ex]
\matrix[
matrix of math nodes,
column sep=0ex,
] (m)
{
	1\cdots 1 &[-0.8ex]0&[-0.7ex] 0 \cdots 0 & 0\cdots 0  \\
	{} & {} &{} & 0 \cdots 0\\
};
\draw
([yshift=-4ex]{$(m-1-3)!.5!(m-1-4)$} |- m.north) -- ([yshift=+1ex]{$(m-1-3)!.5!(m-1-4)$} |- m.south);
\draw
(m-1-1.south west) -- (m-1-4.south east);
\node[preaction={fill=red},inner sep=-1.4pt, anchor=base, rectangle, rounded corners=1.4mm, label={[anchor=center, inner sep=0pt, ]center:${\color{white}0}$}, fit=(m-1-2) (m-1-2)] {};
\end{tikzpicture}
\quad\longleftrightarrow\quad
\begin{tikzpicture}[baseline=-.65ex]
\matrix[
matrix of math nodes,
column sep=0ex,
] (m)
{
	1\cdots 1 &[-0.8ex]1&[-0.7ex] 0 \cdots 0 & 0\cdots 0  \\
	{} & {} &{} & 0 \cdots 0\\
};
\draw
([yshift=-4ex]{$(m-1-3)!.5!(m-1-4)$} |- m.north) -- ([yshift=+1ex]{$(m-1-3)!.5!(m-1-4)$} |- m.south);
\draw
(m-1-1.south west) -- (m-1-4.south east);
\node[preaction={fill=red},inner sep=-1.4pt, anchor=base, rectangle, rounded corners=1.4mm, label={[anchor=center, inner sep=0pt, ]center:${\color{white}1}$}, fit=(m-1-2) (m-1-2)] {};
\end{tikzpicture}\!,
\label{unarytrans}
\end{equation}
changing a $0$ to a $1$ at position $a$.
This edge could be a part of a path that originated in $a-1$ possible vertices $s$
with fewer 1's in the unary clock.
On the other hand, the path could end in $N+1+A-(a-1)$ states, as we have $t>s$.
The edge load for this type of edge is thus $2(a-1)(N+1+A-(a-1))$, with a $2$ for paths from $t$ to $s$.
This is maximized for $a-1=\frac{N+1+A}{2}$,
but we also know that $a<N$. Assuming $A>N$, the maximum within this range of possible $a$'s ia achieved at $a=N$, with the value $2(N-1)(A+2)$. We will see that we will need a larger upper bound  for paths through the other types of edges described bellow.

Second, we could be changing an idling bit
\begin{equation}
\begin{tikzpicture}[baseline=-.65ex]
\matrix[
matrix of math nodes,
column sep=0ex,
] (m)
{
	1\cdots 1 & 1\cdots 1  &  1 & \dots \\
	{}	 & \alpha & 0& \dots \\
};
\draw
([yshift=-4ex]{$(m-1-1)!.5!(m-1-2)$} |- m.north) -- ([yshift=+1ex]{$(m-1-1)!.5!(m-1-2)$} |- m.south);
\draw
(m-1-1.south west) -- (m-1-4.south east);
\draw[densely dashed]
([xshift=-0.1ex,yshift=-0.5ex]m-1-3.north west) -- ([xshift=-0.1ex,yshift=+0.5ex]m-2-3.south west);
\node[preaction={fill=red},inner sep=-1.4pt, anchor=base, rectangle, rounded corners=1.4mm, label={[anchor=center, inner sep=0pt, ]center:${\color{white}0}$}, fit=(m-2-3) (m-2-3)] {};
\end{tikzpicture}
\quad\longleftrightarrow\quad
\begin{tikzpicture}[baseline=-.65ex]
\matrix[
matrix of math nodes,
column sep=0ex,
] (m)
{
	1\cdots 1 & 1\cdots 1  &  1 & \dots \\
	{}	 & \alpha & 1& \dots \\
};
\draw
([yshift=-4ex]{$(m-1-1)!.5!(m-1-2)$} |- m.north) -- ([yshift=+1ex]{$(m-1-1)!.5!(m-1-2)$} |- m.south);
\draw
(m-1-1.south west) -- (m-1-4.south east);
\draw[densely dashed]
([xshift=-0.1ex,yshift=-0.5ex]m-1-3.north west) -- ([xshift=-0.1ex,yshift=+0.5ex]m-2-3.south west);
\node[preaction={fill=red},inner sep=-1.4pt, anchor=base, rectangle, rounded corners=1.4mm, label={[anchor=center, inner sep=0pt, ]center:${\color{white}1}$}, fit=(m-2-3) (m-2-3)] {};
\end{tikzpicture}\!,
\label{idlingtrans}
\end{equation}
or, third, an extra bit
\begin{equation}
\begin{tikzpicture}[baseline=-.65ex]
\matrix[
matrix of math nodes,
column sep=0ex,
] (m)
{
	1\cdots 1 & 1\cdots 1  &  0 & 0\dots0 \\
	{}	 & \alpha & 0& 0\dots0 \\
};
\draw
([yshift=-4ex]{$(m-1-1)!.5!(m-1-2)$} |- m.north) -- ([yshift=+1ex]{$(m-1-1)!.5!(m-1-2)$} |- m.south);
\draw
(m-1-1.south west) -- (m-1-4.south east);
\draw[densely dashed]
([xshift=-0.1ex,yshift=-0.5ex]m-1-3.north west) -- ([xshift=-0.1ex,yshift=+0.5ex]m-2-3.south west);
\node[preaction={fill=red},inner sep=-1.4pt, anchor=base, rectangle, rounded corners=1.4mm, label={[anchor=center, inner sep=0pt, ]center:${\color{white}0}$}, fit=(m-1-3) (m-1-3)] {};
\end{tikzpicture}
\quad\longleftrightarrow\quad
\begin{tikzpicture}[baseline=-.65ex]
\matrix[
matrix of math nodes,
column sep=0ex,
] (m)
{
	1\cdots 1 & 1\cdots 1  &  1 & 0\dots0 \\
	{}	 & \alpha & 0& 0\dots0 \\
};
\draw
([yshift=-4ex]{$(m-1-1)!.5!(m-1-2)$} |- m.north) -- ([yshift=+1ex]{$(m-1-1)!.5!(m-1-2)$} |- m.south);
\draw
(m-1-1.south west) -- (m-1-4.south east);
\draw[densely dashed]
([xshift=-0.1ex,yshift=-0.5ex]m-1-3.north west) -- ([xshift=-0.1ex,yshift=+0.5ex]m-2-3.south west);
\node[preaction={fill=red},inner sep=-1.4pt, anchor=base, rectangle, rounded corners=1.4mm, label={[anchor=center, inner sep=0pt, ]center:${\color{white}1}$}, fit=(m-1-3) (m-1-3)] {};
\end{tikzpicture}\!.\label{extratrans}
\end{equation}
To simplify our expressions, we will only find an upper bound on the number of paths.
Let $a=|\alpha|$. There are not more than $N+1+\left(\sum_{i=1}^{a} 2^{i}\right)+1$
possible starting points (the original unary states plus states with possibly fewer 1's in the extra qubits and different $\alpha$'s).
On the other hand, there are no more than
	$\sum_{j=1}^{C-a} 2^j$
possible endpoints for the paths ($\alpha$ must stay fixed, the rest of the extra and idling bits can change).
Not forgetting a factor of 2 for path symmetry, we thus obtain an upper bound on the number of canonical paths that utilize an edge $(a,b)$:
\begin{align}
\label{IdlingPartEdgeLoad}
	\sum_{\gamma_{s,t} \ni (a,b)} 1
	 \leq
	2 \left( N+2 + \sum_{i=1}^{a} 2^i\right) \left(\sum_{j=1}^{C-a} 2^{j}\right)
	&=2 \left(N+2 + 2^{a+1}-2\right)\left(2^{C-a+1}-2\right) \\
	&\leq 2 \left(N+2\right)\left(2^{C-a+1}-2\right) = 2 \left(N+2\right)A.
\end{align}
as this function is decreasing with $a$, with a maximum at $a=0$.
We can loosen this to $\leq 4N(A+1)$ to get an upper bound for the edge load of all three types of edges.
Recall \eqref{congestion}, and that there are $N+1+A = (1+z)N$ total states, so that $\pi(a) = (N+1+A)^{-1}$ for any state $a$.
We obtain
\begin{align}
	\rho&\leq \frac{2(C+1)4N(A+1)}{N+1+A} = \frac{8z(C+1)N}{z+1}.
\end{align}

The longest canonical path connects $10\cdots0|0\cdots 0|0\cdots 0$ to $1\cdots1|1\cdots 1|1\cdots 1$ and has $l=N+2C<2N$ steps, for $z=\poly(N)$, which gives us $2C= 2\log_2(zN+1)-2<N$ for large enough $N$. Plugging this into \eqref{Pgap} and \eqref{IdlingChainP}, the gaps of the Markov Chain $P$ and $H|_{\textrm{legal}}$ must thus obey
\begin{align}
	\Delta_P &\geq \frac{1}{\rho l} \geq \frac{z+1}{16z(C+1)N^2}, \\
	\Delta_{H|_{\textrm{legal}}}&=2(C+1) \Delta_P
	\geq \frac{z+1}{8zN^2} = \Omega(N^{-2}),
\end{align}
what we wanted to prove.
\end{proof}

Thus, our idling chain construction lets us ``amplify'' the result as much as we want. The ground state can have overlap $\frac{z}{z+1}=1-\epsilon$ with the computation done, while we added only $\log (zN)$ extra qubits.
Moreover, we do not mess up the gap -- it remains $\oOmega{N^{-2}}$, even if $z=\poly(N)$.
Let us compare this to a longer-running domain-wall clock. Amplifying the overlap to $1-\epsilon$ would mean a unary clock of length  $N' = N \epsilon^{-1}$,
while the gap would shrink to $\oOmega{N^{-2} \epsilon^2}$.

\subsection{The second style of idling: Multicog clocks}
\label{sec:cog}

In this Section we describe another way of idling the Feynman clock ``engine''.
In some ways, it is less effective than the idling chain in Section~\ref{sec:idling}
-- it can require more qubits and a higher locality and degree of interactions.
However, it is aimed at a different application -- a dynamical construction. In the legal clock subspace, the dynamics of the evolution with this Hamiltonian are simply a quantum walk on a line, just as for the original unary clock.

\subsubsection{A qutrit surfer on a line}
A unary or {\em domain wall} clock is a progression of states on a line of length $L$ with a single domain wall between 1's and 0's (shown here for $L=5$):
\begin{align}
	\ket{10000},\quad
	\ket{11000},\quad
	\ket{11100},\quad
	\ket{11110}. \label{domainwallstates}
\end{align}
These states form the ground state subspace of the 2-local Hamiltonian
\begin{align}
	H_{\textrm{dwcheck}} = \ket{0}\bra{0}_1 + \ket{1}\bra{1}_{L} + \sum_{i=1}^{L-1} \ket{01}\bra{01}_{i,i+1}.
\end{align}
Let us put another state ``$2$'' on this domain wall and call it a surfer.
A {\em qutrit surfer} on a line of length $L$ is then a linear progression of quantum states (here for $L=5$):
\begin{align}
	\ket{20000},\quad
	\ket{12000},\quad
	\ket{11200},\quad
	\ket{11120},\quad
	\ket{11112}. \label{surferstates}
\end{align}
We can also write down a 2-local Hamiltonian whose ground state subspace is made from these configurations.
It is made from terms on successive pairs of qutrits forbidding domain walls without surfers, as well as double surfers
\begin{align}
	\left(\ket{01}\bra{01} + \ket{10}\bra{10}+\ket{22}\bra{22}\right)_{i,i+1}. \label{surferbadwall}
\end{align}
Next, we need
\begin{align}
	\left(\ket{21}\bra{21}+\ket{02}\bra{02}\right)_{i,i+1}, \label{surferotherwall}
\end{align}
forbidding surfing $\cdots 00\,2\,11 \cdots$ on the other type of domain walls. Finally, we add endpoint terms
\begin{align}
	\ket{0}\bra{0}_1+\ket{1}\bra{1}_L \label{surferend}
\end{align}
that eliminate the boring ``dead'' states $\ket{0\cdots 0}$ and $\ket{1\cdots 1}$.
The whole Hamiltonian is then
\begin{align}
	H^{\text{surf}}_{\text{check}} =
	\ket{0}\bra{0}_1 + \ket{1}\bra{1}_{L}
	+ \sum_{i=1}^{L-1} \left(\ket{01}\bra{01}+\ket{10}\bra{10}+ \ket{21}\bra{21}+\ket{02}\bra{02} +\ket{22}\bra{22}\right)_{i,i+1},
	\label{surferline}
\end{align}
and its zero-energy states\footnote{Note that all of the terms in $H^{\textrm{surf}}_{\text{check}}$ are positive semidefinite; we can make all of them happy at the same time. The frustration free states are given by \eqref{surferstates} -- there is a domain wall somewhere on the line, and a single surfer is riding on it.
} have the form $1 \cdots 1\, 2\, 0\cdots 0$, with a single surfer ``2'' on the domain wall, as in
\eqref{surferstates}.

Let us now add ``dynamics'' of surfer movement (2-local ``rules'') to this model:
\begin{align}
	\cdots 11\,\highlight{20}\, 00 \cdots \,\longleftrightarrow \,\cdots 11\,\highlight{12} \,00\cdots, \label{surfrule}
\end{align}
This rule translates to nearest-neighbor, positive semidefinite Hamiltonian terms
\begin{align}
	H^{\text{surf}}_{\text{dyn}} = \left(\ket{20}-\ket{12}\right)\left(\bra{20}-\bra{12}\right)_{i,i+1}. \label{surfright}
\end{align}
The unique (and still frustration-free) ground state of the surfer model (with the dynamics) is then the uniform superposition over all surfers \eqref{surferstates}
\begin{align}
	\ket{\psi_0} = \sum_{t=1}^{L} \ket{t_{\text{surf}}} = \sum_{t=1}^{L} \ket{ 1\cdots 1 2_t 0 \cdots 0 }.
\end{align}
It is simple to analyze the spectrum of a surfer model -- in the ``legal'' subspace (single-surfer) it is a quantum walk on a line
\begin{align}
	\left( H^{\text{surf}}_{\text{check}} + H^{\text{surf}}_{\text{dyn}} \right)\Big|_{\text{legal}}
	= \sum_{t=1}^{L-1} \left(\ket{t+1}-\ket{t}\right)\left(\bra{t+1}-\bra{t}\right),
\end{align}
which is a tridiagonal matrix with gap \cite{KomaNachtergaele}
\begin{align}
	\Delta_{\text{surfer}}^{\text{line}} = \oo{\frac{1}{L^2}}.
	\label{surferlinegap}
\end{align}
Now the ``illegal'' subspaces contain at least one state directly detected by $H^{\text{surf}}_{\text{check}}$, with constant energy away from 0.
Therefore, \eqref{surferlinegap} is the gap of $H^{\text{surf}}_{\text{check}} + H^{\text{surf}}_{\text{dyn}}$; it scales like $\oo{L^{-2}}$.

\subsection{A qutrit surfer on a cycle}

Let us now wrap the line with the surfer around and put the surfer on a cycle. In contrast to the line, we now allow
the surfer to ride on both types of domain walls, i.e.
\begin{align}
		\cdots11\,2\,00\cdots, \qquad \textrm{and} \qquad
		\cdots00\,2\,11\cdots \label{walls0110}.
\end{align}
We will keep the basic terms $\ket{01}\bra{01}+ \ket{10}\bra{10} +\ket{22}\bra{22}$ from \eqref{surferline} implying surfers sitting only on domain walls. First, we add the dynamics \eqref{surfright} of the surfer riding the domain wall $\cdots 1 2 0 \cdots$.
Second, we also allow it to move on the other type of domain wall as
\begin{align}
	\cdots 00\,\highlight{21}\,11 \cdots \,\longleftrightarrow \,\cdots 00\,\highlight{02}\,11\cdots,
\end{align}
which translates to a Hamiltonian term
\begin{align}
	\left(\ket{02}-\ket{21}\right)\left(\bra{02}-\bra{21}\right)_{i,i+1}. \label{surfleft}
\end{align}
Instead of constraining the endpoints with \eqref{surferend}, we now designate a special interaction for the sites $1$ and $L$ (they sit next to each other when we wrap the line into a cycle). The interaction will consist of
\begin{align}
	\left(\ket{00}\bra{00}+\ket{11}\bra{11}+\ket{22}\bra{22}\right)_{L,1},
	\label{nonequal}
\end{align}
making sure the endpoints are not equal, as well as special dynamics across the endpoints
\begin{align}
	\cdots 11\, \highlight{2\textrm{\textbar}1}\, 11\cdots \,\longleftrightarrow \,\cdots 11\, \highlight{1\textrm{\textbar}2}\, 11\cdots, \\
	\cdots 00\, \highlight{2\textrm{\textbar}0}\, 00\cdots \,\longleftrightarrow \,\cdots 00\, \highlight{0\textrm{\textbar}2}\, 00\cdots,
\end{align}
expressed as the Hamiltonian terms
\begin{align}
	\left(\ket{21}-\ket{12}\right)\left(\bra{21}-\bra{12}\right)_{L,1}
	+\left(\ket{20}-\ket{02}\right)\left(\bra{20}-\bra{02}\right)_{L,1}.
	\label{surfspecial}
\end{align}
The complete surfer-cycle Hamiltonian is then $H^{\text{surfcycle}}_{\text{check}} + H^{\text{surfcycle}}_{\text{dyn}}$ with
\begin{align}
	H^{\text{surfcycle}}_{\text{check}} &=
		\left(\ket{00}\bra{00}+\ket{11}\bra{11}+\ket{22}\bra{22}\right)_{L,1} +
		\sum_{i=1}^{L-1} \left(\ket{01}\bra{01}+ \ket{10}\bra{10} +\ket{22}\bra{22}\right)_{i,i+1}\label{surferbadwallcycle}	\\
	H^{\text{surfcycle}}_{\text{dyn}} &=
	\sum_{i=1}^{L-1} \big(
		\left(\ket{12}-\ket{20}\right)\left(\bra{12}-\bra{20}\right)
		+ \left(\ket{02}-\ket{21}\right)\left(\bra{02}-\bra{21}\right)
		\big)_{i,i+1} \label{surfcycledynamics}\\
	&+ \big(
		\left(\ket{12}-\ket{21}\right)\left(\bra{12}-\bra{21}\right)
	 	+ \left(\ket{02}-\ket{20}\right)\left(\bra{02}-\bra{20}\right)
		\big)_{L,1}. \nonumber
\end{align}

Let us look at what the allowed states are now. Because of the condition \eqref{nonequal}, there has to be at least one domain wall in the system. Because of \eqref{surferbadwallcycle}, there has to be a surfer on this wall (and thus at least one surfer in the system). Finally, the dynamical terms \eqref{surfcycledynamics}, imply that surfers must exist in superpositions of their movement about the cycle. However, two surfers should not appear next to each other because of the last term in \eqref{surferbadwallcycle}. Thus, there is a unique frustration-free ground state of the qutrit cycle Hamiltonian: the uniform superposition of all single-surfer states depicted in Figure~\ref{fig:singlecog}.

\begin{figure}
\begin{center}
\includegraphics[width=17cm]{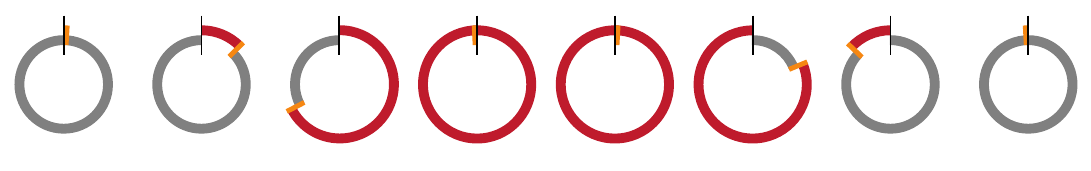}%
\caption{The (progression of) allowed clock states of a single qutrit surfer cycle. Site $1$ is to the right of the vertical line, site $L$ is to the left of it.}%
\label{fig:singlecog}%
\end{center}
\end{figure}

Let us calculate the gap. First, in the single-surfer, proper-domain-wall subspace, the Hamiltonian is a quantum walk on a cycle of length $2L$, with gap $\oo{L^{-2}}$ \cite{QWalksActaPhysica}.

Second, the surfers as well as the $L,1$ boundary ``count'' the number of domain walls. We are on a cycle, so the number of such jumps needs to be even (and the number of surfers needs to be odd) for a possible zero-energy state.
Thus all even-surfer subspaces (including the no-surfer subspace) are lower-bounded in energy by a constant, as they are made from states immediately detected by the terms \eqref{surferbadwallcycle} and \eqref{nonequal}.

What about the subspaces with an odd number of surfers? We can forget about the states with bad domain walls (across $L,1$ or of the $01$ and $10$ type inside) as their energy is also at least a constant.
Thus, we are left with proper states of two possible types:
\begin{align}
		\cdots 11 \,2\, 00\cdots 00 \, 2 \, 11|00\cdots 00\,2\,11 \cdots \\
		\cdots 00 \,2\, 11\cdots 11 \, 2 \, 00|11\cdots 11\,2\,00 \cdots
\end{align}
We will now play the term $H_{22}=\ket{22}\bra{22}$ against the rest of the terms $H_r$.
There is no common 0-energy eigenstate for both of the positive semidefinite terms $H_{22}$ (it hates the 22 sequence) and $H_{r}$ (it demands uniform superpositions over all possible states that the surfers can move to). We now use a geometric lemma \cite{KitaevBook}

\begin{align}
	\lambda^{H_{22} + H_{r}}_0 \geq \min\left\{\lambda_1^{H_{22}},\lambda_1^{H_{r}}\right\} \sin^2 \frac{\theta}{2},
\end{align}
where $\theta$ is the angle between the zero-energy subspaces of these terms.
The ground state of $H_{r}$ in a subspace with $k$ surfers is a uniform superposition over all of their positions (with the rest properly filled with 0s and 1s), and the lowest excited state has energy at least $\lambda_1^{H_r}= \Omega\left(\frac{1}{L^2}\right)$, which we know from solving the quantum walk on a cycle \cite{QWalksActaPhysica}.
The ground state of $H_{22}$ is made from (product) states with no 2 surfers next to each other, and
$\lambda_1^{H_{22}}\geq 1$.
We can get a lower bound on the angle $\theta$ from counting the number of detected states ($\cdots 22\cdots$) in the uniform superposition of well-walled states in the $k$-surfer subspace. There are $p = \binom{L}{k}$ possible positions of $k$ surfers. Out of these, at most
\begin{align}
	\frac{L (L-2) (L-4) \dots  (L-2k+2)}{k!}
\end{align}
arrangements do not have two adjacent surfers.
Let us estimate the ratio of these two numbers.
\begin{align}
	\frac{\#\textrm{no}_{22}}{\#\textrm{all}} &\leq \frac{L (L-2) (L-4) \dots  (L-2k+2)}{L(L-1)(L-2)\dots (L-k+1)}
				\leq \frac{L-2}{L-1}  = 1 - \frac{1}{L-1}.
\end{align}
with the upper bound coming from $k=2$.
Looking at the angle between the null subspaces, we get
\begin{align}
	\cos \theta &= \frac{\sum_{x\in \textrm{all}}\sum_{y\in \textrm{no}_{22}}\braket{x}{y}}{\sqrt{\#\textrm{all}}\sqrt{\#\textrm{no}_{22}}}
		= \sqrt{\frac{\#\textrm{no}_{22}}{\#\textrm{all}}}
	\leq \sqrt{1-\frac{1}{L-1}} \\
	\sin^2 \frac{\theta}{2} &= \frac{1-\cos \theta}{2} \geq \frac{1}{4(L-1)} = \oo{L^{-1}}.
\end{align}
Therefore, we can put a lower bound on the lowest eigenvalue of $H_{r}+H_{22}$ in the subspaces with $k = 2m+1$ surfers, which will be a lower bound on the gap of the whole surfer cycle Hamiltonian:
\begin{align}
	\Delta_{surfer}^{cycle} \geq \lambda^{H_{22}+H_r}_0 = \Omega\left(\frac{1}{L^3}\right).
	\label{surfercyclegap}
\end{align}

\subsection{A multi-cog clock made from several qutrit surfer cycles}

We will now build a clock from multiple coupled qutrit surfer cycles -- {\em cogs}.
Synchronization is not that difficult: the next cog can progress only at the moment when the previous cog has finished its two revolutions (reminiscent of binary addition).
We illustrate the synchronization of $C$ cogs of length $L$ in Figure~\ref{fig:synchrocog}.

\begin{figure}%
\begin{center}
\includegraphics[width=17cm]{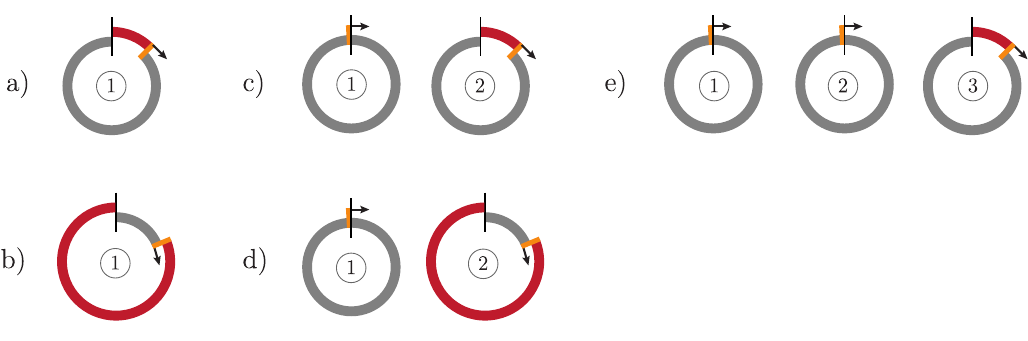}%
\caption{Synchronizing multiple cogs. a) Cog 1 simply progresses through its first revolution. b) Cog 1 simply progresses through its second revolution.
c) When the first cog has finished two revolutions, its transition is coupled to a progression of the second cog.
d) Note that cog 2 could also be in its second revolution.
e) When the second cog has finished two revolutions, {\em and} the first cog has finished two revolutions, they transition together and cog 3 progresses as well.}%
\label{fig:synchrocog}%
\end{center}
\end{figure}

First, let us see what happens for 2 cogs.
We will add an interaction term coupling the end of the second revolution transition of cog 1 with a simple transition of cog 2 as
\begin{align}
	2|0^{(1)}_{L,1} \,\,20^{(2)}_{i,i+1} \,\longleftrightarrow \,
	0|2^{(1)}_{L,1} \,\,12^{(2)}_{i,i+1}
\end{align}
when the second cog is in its first revolution (first line in Figure~\ref{fig:synchrocog}b),
as well as finishing the first revolution of the second cog with
\begin{align}
	2|0^{(1)}_{L,1} \,\,2|1^{(2)}_{L,1} \,\longleftrightarrow \,
	0|2^{(1)}_{L,1} \,\,1|2^{(2)}_{L,1},
\end{align}
and finally
\begin{align}
	2|0^{(1)}_{L,1} \,\,21^{(2)}_{i,i+1} \,\longleftrightarrow \,
	0|2^{(1)}_{L,1} \,\,02^{(2)}_{i,i+1},
\end{align}
when the second cog is in its second revolution (second line in Figure~\ref{fig:synchrocog}b).
The interaction is 4-local.
These transitions need to be rewritten to Hamiltonian terms
just as we rewrote \eqref{surfrule} to \eqref{surfright}.

What will happen when the second cog finishes its second revolution? We're ready to advance a third cog.
If we had 3 cogs, we would need another synchronization interaction to couple the simultaneous end of the second revolution of cogs 1 and 2 with a simple transition of cog 3. First, we have
\begin{align}
	2|0^{(1)}_{L,1} \,\, 2|0^{(2)}_{L,1} \,\, 20^{(3)}_{i,i+1} \,\longleftrightarrow \,
	0|2^{(1)}_{L,1} \,\, 0|2^{(2)}_{L,1} \,\, \,\,12^{(3)}_{i,i+1},
\end{align}
restarting the first two cogs and advancing cog 3. This interaction is 6-local.
Of course, we need a special term for when cog 3 is just finishing its first revolution,
\begin{align}
	2|0^{(1)}_{L,1} \,\, 2|0^{(2)}_{L,1} \,\, 2|1^{(3)}_{L,1} \,\longleftrightarrow \,
	0|2^{(1)}_{L,1} \,\, 0|2^{(2)}_{L,1} \,\, \,\,1|2^{(3)}_{L,1},
\end{align}
and a class of terms for when cog 3 is inside its second revolution:
\begin{align}
	2|0^{(1)}_{L,1} \,\, 2|0^{(2)}_{L,1} \,\, 21^{(3)}_{i,i+1} \,\longleftrightarrow \,
	0|2^{(1)}_{L,1} \,\, 0|2^{(2)}_{L,1} \,\, \,\,02^{(3)}_{i,i+1}.
\end{align}
All of these 3-cog synchronization terms are 6-local.

Following this line of thinking, we can enforce $C$-cog synchronization with $2C$-local terms.

\subsection{The cost of a clock with $C$ cogs of length $L$.}

How long can such a clock run? For $C$ cogs of length $L$, the number of available timesteps (if the last cog does not keep revolving over but stops after its second revolution) is
\begin{align}
	N = (2L)^C.
\end{align}
As noted above, we need $2C$-local terms for $C$-cog synchronization. This is OK for constant $C$. Thus, when we use $C$ cogs, each needs to have length $L=\frac{1}{2}N^{\frac{1}{C}}$ and we are using $CL = \frac{C}{2}N^{\frac{1}{C}}$ qutrits for this. The space requirement is thus a constant root of $N$, i.e. sublinear space, which was our plan.

How costly are the interactions? We require $2C$-local synchronization interactions. The particles that are involved the most are the $L,1$ qutrits of the first cog. They appear in synchronization interactions with all of the cogs.
Their degree (number of interactions) is
\begin{align}
	2L_{4-local} + 2L_{6-local} + \dots + 2L_{2C-local} \approx 2LC = CN^{\frac{1}{C}}.
\end{align}

Let us calculate the gap.
In the proper clock subspace (one surfer in each cog), it will be
\begin{align}
	\Delta_{multicog} = \oo{\frac{1}{N^2}},
	\label{multicoggap}
\end{align}
as the Hamiltonian in that subspace is just a quantum walk on a cycle of length $N$.
The next lowest energy subspace will be one with a {\em single} cog that is messed up (having several surfers). We know a lower bound on its energy from the 2-surfer cog lower bound \eqref{surfercyclegap}. It is $\oOmega{L^{-3}}$, which implies a lower bound at least $\oOmega{N^{-3/C}}$ for the messed-up subspaces. This is larger than \eqref{multicoggap}, and thus not important.

\subsection{Computation and interaction with data}

So far, no interaction with any data was put in. If we want to compute with our multi-cog clock, we need interactions of every cog clock pair (for addressing which unitary to apply) + the data, so $2C+2$-local interactions (simplest counting, using 2-local unitaries).

The simplest implementation would have just 2 cogs of length $\frac{1}{2}\sqrt{N}$, together being able to run for time $N$. The interactions would be 4-local for 2-cog synchronization, 6-local for implementing 2-qubit gates, and 5-local for initialization and readout terms.

\subsection{Doing something like this with qubits}

What if we wanted to do something similar by qubits? It is entirely possible with domain walls, just the transitions would be 4-local. We would introduce terms that hate two domain walls near each other, i.e. $101$ or $010$. Across the points $L,1$, we could use a transition like $10|11\,\leftrightarrow\,
11|01$ that would change a $\cdots 111000|111111\cdots$ turn (with a string of 1's ``growing'' into a string of 0's) of the cog into a
$\cdots 111111|000111\cdots$ turn (with a string of 0's ``growing'' into a string of 1's). The transition $01|00\,\leftrightarrow\,
00|10$ could then be the one that would synchronize with the next cog.
The synchronization would then be $4C$-local. Everything would remain frustration-free and would still work. The gaps would remain $\oOmega{1/\poly(L)}$ or $\oOmega{1/\poly(N)}$. Really, the price we pay is just locality of interactions.

However, what if we allowed for frustration? we could use a pulse-clock then.

\subsection{Idling with a multicog clock}

Let's say we're amplifying a computation with a multicog clock and we want to
add $A=\oo{N}$ extra clock states.
For the idling, we can choose to append a 2-cog clock with $L=\sqrt{N}$.
The cost (blowup) in size is $\oo{\sqrt{N}}$ qutrits (we now also know how to do it with only qubits). Of course, when we use a 2-cog clock for the computation as well, so the overall number of clock states we need remains $\oo{\sqrt{N}}$. On the other hand, the gap of the Hamiltonian still scales like $\oo{N^{-2}}$, as the Hamiltonian in the legal clock subspace is just a walk on a cycle of length $\oo{N}$. On the other hand, illegal clock states are immediately detected by terms in the Hamiltonian, so they have at least constant energy (and there are no transitions from legal to illegal clock states).

\section{Pulse clocks: lower locality that requires tuning.
}
\label{sec:tuning}

Finally, in this Section we present our last result: how to tune a pulse clock (see Section~\ref{sec:pulse}) by frustration. Using terms that are not positive semidefinite, we can choose to energetically prefer a subspace with a single excitation. This way a pulse clock can also be used in constructions for QMA-complete problems, improving on locality of the required terms, while paying the price of a smaller gap (after rescaling the Hamiltonian to have norm-1 terms).

This calculation can be useful for constructions in Hamiltonian complexity, where 2-local interactions in the clock register are needed, but one also needs to ensure the proper clock states are selected. A weaker version of the calculation can be found in Section 4.2 of the paper \cite{GottesmanHastings}, where Gottesman \& Hastings investigate the possible dependence of the entanglement entropy on the gap of a qudit chain. They also add a tuning Hamiltonian with a gap $\oOmega{N^{-4}}$ that prefers the single-excitation subspace. It means their whole Hamiltonian has a $\oOmega{N^{-4}}$ gap. Meanwhile, the tuning Hamiltonian we present here has a gap $\oOmega{N^{-3}}$ between the ground state energy and the lowest energy states from subspaces with $z\neq 1$ excitations. Unfortunately, the better bound on the gap of the tuning Hamiltonian doesn't improve the lower bound on the gap of the whole Gottesman-Hastings Hamiltonian, whose scaling is then governed by other terms, already within the 1-excitation subspace.

Recall from Section~\ref{sec:pulse} that a {\em pulse} clock is a progression of states on a line with a single excitation
$\ket{1000\cdots}$, $\ket{0100\cdots}$, etc., connected to each other by transitions governed by the 2-local Hamiltonian $H_N^{\textrm{pulse,L}}$ \eqref{Hpulseproj}.
This positive semidefinite Hamiltonian conserves the number of excitations (the 1's) and thus has $N+1$ invariant subspaces with a fixed number of 1's.
The Hamiltonian energetically prefers symmetry, so that each uniform superposition of states with a particular number of $1$'s is a zero energy, frustration-free ground state of $H_N^{\textrm{pulse,L}}$.
We would like $\ket{\tilde{1}}$, the uniform superposition of states with a single 1 to be the unique ground state. We will achieve this by adding local terms to $H_N^{\textrm{pulse,L}}$.

\begin{theorem}[Pulse clock tuning]
There exists a 2-local, nearest neighbor Hamiltonian on a chain of length $N$, whose unique ground state is $\ket{\tilde{1}} = \frac{1}{\sqrt{N}}\left(\ket{1000\cdots}+\ket{0100\cdots}
+\ket{0010\cdots}+\dots\right)$, the uniform superposition of states with a single excitation, it has ground state energy zero, and and whose gap is $\Delta_{\tilde{H}} = \Omega(N^{-3})$.
\end{theorem}

\begin{proof}
We start with the 2-local, nearest-neighbor pulse clock, Laplacian-type Hamiltonian \eqref{Hpulseproj} on a chain, whose $N+1$ ground states are $\ket{\tilde{z}}$, the uniform superpositions of states with exactly $z$ ones for $Z = 0,\dots, N$.
\begin{align}
	H_{N}^{\textrm{pulse,L}} &= \sum_{x=1}^{N}
	\left(\ket{01}-\ket{10}\right)\left(\bra{01}-\bra{10}\right)_{x,x+1}. \label{Hpulseproj2}
\end{align}
Our goal is to raise the energy of the uniform superpositions with $z\neq 1$ of $1$'s. In particular, we need to deal with the {\em dead} subspace that contains only the product state $\ket{0\cdots 0}$, as well as the higher-$z$ subspaces.
Let us add
\begin{align}
	H^{\textrm{tuning}}_N = V\, \ii - V \sum_{x=1}^N \ket{1}\bra{1}_x + \sum_{x=1}^{N-1} \ket{11}\bra{11}_{x,x+1}, \label{Htune}
\end{align}
another excitation-number preserving Hamiltonian.
The first term is a constant shift. The second term prefers excitations. The third term energetically punishes excitations that sit next to each other. Let us analyze
\begin{align}
	\tilde{H} = H_N^{\textrm{pulse,L}} + H^{\textrm{tuning}}_N \label{Htilde}
\end{align}
in each of its invariant subspaces labeled by the number of excitations $z$.
First, there is only one state $\ket{00\cdots 0}$ that lives in the 0-excitation subspace. The Hamiltonian $H^{\textrm{tuning}}_N$ gives it energy $V$.

Second, the Hamiltonian $H^{\textrm{tuning}}_N$ is diagonal for all states from the single-excitation subspace, giving them energy $V-V=0$. The term $H_N^{\textrm{pulse,L}}$ dictates that the ground state there is $\ket{\tilde{1}}$, with energy 0.

Third, let us look at subspaces with $z>1$ excitations and find a lower bound on the lowest eigenvalue of $H^{\textrm{pulse,L}}_N+H^{\textrm{tuning}}_N$ in each such subspace.
For a specific $z$, let us add a shift to the Hamiltonian and view it as a sum of two positive semidefinite terms $A$ and $B$:
\begin{align}
	H'_z &= \tilde{H} + (z-1)V \ii = \underbrace{H^{\textrm{pulse,L}}_N}_{A}
	+ \underbrace{H^{\textrm{tuning}}_N + (z-1)V \ii}_{B}. \label{HAB}
\end{align}
Restricting $B$ to the subspace $\mathcal{H}_z$ with $z$ excitations gives us
\begin{align}
	B \big|_{\mathcal{H}_z} = \sum_{x=1}^{N-1}\ket{11}\bra{11}_{x,x+1},
\end{align}
with the shift canceling the contribution of the first two terms in \eqref{Htune}.
Therefore, $B$ restricted to $\mathcal{H}_z$ is a sum of projectors and thus positive semidefinite. This lets us use a geometric lemma \cite{KitaevBook} about the lowest eigenvalue of a sum of two positive semidefinite operators $A$ and $B$ whose ground state subspaces do not intersect. It says that there is a lower bound
\begin{align}
	\lambda^{A+B}_0 \geq \min\left\{\lambda_1^{A},\lambda_1^{B}\right\} \sin^2 \frac{\theta}{2}
\end{align}
on the lowest eigenvalue of a sum of positive semidefinite operators $A+B$,
in terms of the second lowest eigenvalues of $A$ and $B$, and the angle $\theta$ between the ground state subspaces of $A$ and $B$.

The ground state of $A$ in $\mathcal{H}_z$ is the uniform superposition over all states with $z$ excitations; the first excited state for $A$ has energy $\lambda^A_1 = \Theta\left(\frac{1}{L^2}\right)$, as we know from the gap of the Heisenberg model \cite{KomaNachtergaele}.
The ground states of $B$ are states with $z$ excitations that do not have any ``11'' substrings. The lowest excited state of $B$ has energy $\lambda^B_1 = 1$ for $z\leq \frac{N+1}{2}$ and even more
for $z>\frac{N+1}{2}$.

Let us now calculate $\theta$, the angle between the null subspaces of $A$ and $B$.
The ground state subspace of $A$ is a single state -- the symmetric superposition $\ket{\tilde{z}}$ with $z$ 1's.
The vector from the ground subspace of $B$ with the largest overlap with $\ket{\tilde{z}}$
is the uniform superposition of all the states from the ground subspace of $B$, as these states all appear in $\ket{\tilde{z}}$. We thus only need to count their number $\#\textrm{no}_{11}$, and express
\begin{align}
	\cos \theta_z = \braket{\tilde{z}}{\tilde{z}_{\textrm{without }11\textrm{'s}}}
	= \sqrt{\frac{\#\textrm{no}_{11}}{\#\textrm{all}}}. \label{costheta}
\end{align}

There are $\#\textrm{all} = \binom{N}{z}$ strings with $z$ ones. What is the number $\#\textrm{no}_{11}$ of $N$-bit strings with $z$ ones and no $11$ substrings?
It makes sense to count them only for $z\leq \frac{N+1}{2}$, as above that $\#\textrm{no}_{11}=0$. All such strings must have a ``backbone'' --
a collection of $z$ substrings (10, 10, \dots, 10, 1) creating $z+1$ bins
\begin{align}
	\cdots |_{10}| \cdots |_{10}| \cdots |_{10}| \cdots |_{10}| \cdots |_{10}| \cdots |_{1}| \cdots
\end{align}
into which we need to distribute the remaining $N-(2z-1)$ zeros.
Therefore,
\begin{align}
	\#\textrm{no}_{11} &= \binom{ N-(2z-1)+z}{z} = \binom{ N-z+1}{z}.
\end{align}
For $\frac{N+1}{2}> z\geq 2$ this implies
\begin{align}
	\cos \theta &= \sqrt{\frac{\#\textrm{no}_{11}}{\#\textrm{all}}}
	= \sqrt{\frac{\binom{N -z + 1}{z}}{\binom{N}{z}}}
	= \sqrt{\frac{(N-z)}{N}\frac{(N-z-1)}{(N-1)}\dots\frac{(N-z-(z-2))}{
	(N-(z-2))}} \nonumber\\
	& \leq \sqrt{\frac{N+2-2z}{N+2-z}} = \sqrt{1-\frac{z}{N+2-z}}
	\leq \sqrt{1-\frac{z}{N}} \leq 1-\frac{z}{2N}.  \label{cosbound}
\end{align}
This upper bound gives us a lower bound on $\sin^2 \frac{\theta}{2} = \frac{1}{2}(1-\cos \theta) \geq \frac{z}{4N}$.
Plugging everything into the geometric lemma \cite{KitaevBook}, for $\frac{N+1}{2}> z\geq 2$ we obtain
\begin{align}
	\lambda^{A+B}_0 \geq \min\left\{\lambda_1^{A},\lambda_1^{B}\right\} \sin^2 \frac{\theta}{2}
	= \oOmega{N^{-2}} \frac{z}{4N} = \oOmega{z N^{-3}}.
\end{align}
Turning back to our Hamiltonian $\tilde{H}$ before the shift \eqref{Htilde}, we find that the ground state energy in each of the $\frac{N+1}{2}> z\geq 2$ subspaces is
\begin{align}
	E_z = \lambda^{A+B}_0 - (z-1)V \geq \Omega(z N^{-3}) - (z-1)V = \Omega(z N^{-3}),
\end{align}
if we choose
\begin{align}
	V = N^{-3}. \label{Vchoice}
\end{align}

Meanwhile, for $z\geq \frac{N+1}{2}$ we have $\cos \theta = 0$ and $\sin^2 \frac{\theta}{2} = 1$
and $\lambda^{A+B}_0 = \Omega(N^{-2})$
and $E_z = \lambda^{A+B}_0 - (z-1)V \geq \Omega(N^{-2}) - (z-1)V = \Omega(N^{-2})$, for the choice \eqref{Vchoice}. Finally, recall that for $z=0$ we had $E_0 = V = \oo{N^{-3}}$.

Therefore, choosing $V=N^{-3}$ in the tuning term \eqref{Htune}, adding it to the pulse clock, Laplacian-type Hamiltonian \eqref{Hpulseproj2} does what we wanted to prove. The state $\ket{\tilde{1}} = \frac{1}{\sqrt{N}} \left(\ket{100\cdots} + \ket{010\cdots} + \dots\right)$ is the unique, zero-energy ground state of
\eqref{Htilde}, and the gap of $\tilde{H}$ is $\Delta_{\tilde{H}}= \Omega(N^{-3})$.

Note that the improvement over the calculation in Eqn. (73) in \cite{GottesmanHastings} comes from \eqref{cosbound}. There, we add a factor of $z$, helping us increase their gap lower bound from $\Omega(N^{-4})$ to $\Omega(N^{-3})$.

However, this result may possibly be further improved, as numerical investigation of $V=N^{-\frac{3}{2}}$ indicate a gap scaling as $\Theta\left(N^{-2}\right)$.

\end{proof}

\section{Acknowledgements}
LC has done the work on this paper during his PhD studies at the Faculty of Mathematics, Physics and Informatics of the  Comenius University in Bratislava. DN's research has received funding from the People Programme (Marie Curie Actions) EU's 7th Framework Programme under REA grant agreement No. 609427. This research has been further co-funded by the Slovak Academy of Sciences.
LC and DN were further supported by the Slovak Research and Development Agency grant QETWORK APVV-14-0878. DN also acknowledges funding from the project OAQS VEGA 2/0130/15.
ZL acknowledges support by ARO Grant W911NF-12-1-0541, NSF Grant CCF-1410022 and Templeton Foundation Grant 52536.
ZL and DN also thank the Simons Institute for the Theory of Computing in Berkeley, where parts of this work were carried out, for their hospitality.


\bibliographystyle{abbrv}	
\bibliography{TheClockPaper}

\end{document}